\documentclass[11pt]{article}
\usepackage{tikz}
\usepackage{authblk}
\usepackage{geometry}
\usepackage[english]{babel}
\usepackage[utf8]{inputenc}
\usepackage[T1]{fontenc}
\usepackage{indentfirst}
\usepackage{amsmath}
\usepackage{amssymb}
\usepackage{amsthm}
\usepackage{proof}
\usepackage{eufrak}
\usepackage{graphicx}
\usepackage{psfrag}

\allowhyphens

\newcommand{\paulcite}{\cite{fendley-fermions-in-disguise} }
\newcommand{\paulcitee}{\cite{fendley-fermions-in-disguise}}

\newcommand{\HH}{\mathcal{H}}
\newcommand{\CC}{\mathcal{C}}
\newcommand{\VV}{\mathcal{V}}
\newcommand{\SSS}{\mathcal{S}}
\newcommand{\ZZ}{\mathcal{Z}}
\newcommand{\complex}{\mathbb{C}}

\newcommand{\valos}{\mathbb{R}}
\newcommand{\eps}{\varepsilon}

\newtheorem{thm}{Theorem}
\newtheorem{prop}{Proposition}
\newtheorem{conj}{Conjecture}

\newcommand{\ket}[1]{{\left|#1\right\rangle}}
\newcommand{\bra}[1]{{\left\langle #1\right|}}


\usepackage{ifpdf}

\ifpdf
\usepackage{epstopdf}
\usepackage[pdftex,colorlinks,urlcolor=blue,citecolor=blue,linkcolor=blue]{hyperref}
\else
\usepackage[hypertex,colorlinks,urlcolor=blue,citecolor=blue,linkcolor=blue]{hyperref}
\fi
\begin{document}
\numberwithin{equation}{section}

\title{Quantum circuits with free fermions in disguise}
\author[1]{Kohei Fukai}
\author[2]{Bal\'azs Pozsgay \thanks{corresponding author, pozsgay.balazs@ttk.elte.hu}}

\affil[1]{Department of Physics, Graduate School of Science, \authorcr The University of Tokyo, \authorcr7-3-1, Hongo, Bunkyo-ku, Tokyo, 113-0033, Japan}
\affil[2]{MTA-ELTE ``Momentum'' Integrable Quantum Dynamics Research Group, \authorcr ELTE Eötvös Loránd University,
  \authorcr P\'azm\'any P. s\'et\'any 1/A, H-1117 Budapest, Hungary}

\maketitle

\begin{abstract}
  Recently multiple families of spin chain models were found, which have a free fermionic spectrum,
  even though they are not solvable by a Jordan-Wigner transformation. Instead, the free fermions emerge as a result of
  a rather intricate construction.
  In this work we consider the quantum
  circuit formulation of the problem. We construct circuits using local unitary gates built from the terms in the local Hamiltonians
  of the respective models, and ask the question:  which circuit geometries (sequence of gates) lead to a free fermionic
  spectrum?
  Our main example is the 4-fermion model of
  Fendley, where we construct free fermionic circuits with various geometries. In certain cases we prove the
  free fermionic nature, while for other geometries we confirm it numerically. Surprisingly, we find that many standard
  brickwork circuits are not free fermionic, but we identify certain symmetric constructions which are. 
  We also consider a recent generalization of the 4-fermion model and obtain the factorization of its transfer matrix, {and subsequently derive a free-fermionic circuit for this case as well.}
\end{abstract}

\section{Introduction}

In many body physics the free theories are the simplest ones: in the absence of interactions practically every physical
quantity can be 
computed with ease, knowing the solution of the one body problem. Sometimes certain models might appear interacting, but
it might be possible to express the local 
variables/fields with free modes/fields. The most famous example for such a transformation is the Jordan-Wigner transformation
\cite{jordan-wigner}, which expresses local spin variables of a spin-1/2 chain using fermionic operators, thereby making
many models solvable by free fermions \cite{XX-original,Schulz-Mattis-Lieb}.

Given the success of the Jordan-Wigner transformation it is very natural to ask: For which spin systems can it be applied,
under what circumstances? And what other types of transformations can exist that map seemingly interacting
operators/fields to completely free ones? If other mappings are found in other situations, then which physical
quantities can be computed with those methods? 

The first question was answered recently in the papers \cite{chapman-jw,japan-free-fermion-JW}, where graph theoretical
criteria were found for the 
applicability of any kind of generalized Jordan-Wigner transformation. The other two questions were considered in
multiple papers in the last couple of years, and this area of research is still in development.

In the seminal paper \paulcite Fendley introduced a family of spin chain models (``free fermions in disguise''), where
the Hamiltonian has 
4-fermion interactions, but there exist hidden free fermionic creation/annihilation operators which eventually diagonalize the
problem. The true free fermions are given by an intricate construction, involving the original interacting fermionic
variables. Distinguished properties of the models are that they are free for every choice of the coefficients in the
Hamiltonian (including disorder), and that every energy level is degenerate with the same degeneracy that grows
exponentially with the volume.
Generalizations of the 4-fermion model appeared in \cite{alcaraz-medium-fermion-1,alcaraz-medium-fermion-2}, and 
several physical properties of these new models were
analyzed afterwards in  \cite{alcaraz-medium-fermion-3,rodrigo-ising-and-ffd,rodrigo-random-ffd}.

A graph theoretical understanding for the ``free fermions behind the disguise'' was given in
\cite{fermions-behind-the-disguise}. Finally, a unified graph theoretical framework was presented in
\cite{unified-graph-th}, which also covers the models solvable by the Jordan-Wigner transformation. Both of these works
consider spin-1/2 models where it is specified which terms can be added to the Hamiltonian, but afterwards the models
are free for 
any choice of the coupling constants.
The two works \cite{fermions-behind-the-disguise,unified-graph-th} give
sufficient conditions for the existence of the disguised
free fermions.
Applications to quantum information theory were considered in
\cite{free-fermion-subsystem-codes}, and an extension to free parafermions \cite{free-parafermion} appeared recently in \cite{free-parafermionic-graphs}.

The previously mentioned results do not
cover all possibilities for hidden free fermions. For example, there exist models, where some fine tuning is necessary between
certain coupling constants. One family of such models was discovered in \cite{sajat-FP-model}; the new family appears to
interpolate between models treated earlier in \cite{cooper-anyon,gyuri-susy-1,gyuri-susy-2}, which appeared unrelated at
that time.
The model of
\cite{sajat-FP-model}  requires some fine tuning of its parameters, still allowing for a large freedom of their
choices. 

\bigskip

All of these works deal with spin models defined through their Hermitian Hamiltonians. In contrast, we ask
the question of how to construct unitary quantum circuits with the hidden free fermionic structures.
There are two main reasons for the study of this question.

The first motivation is to use the new models for practical calculations in many-body systems.  We remind that in earlier literature
the Jordan-Wigner solvable models proved to be very useful for the study of non-equilibrium 
dynamics of many body systems, and this extends to the unitary quantum circuits.
Floquet time evolution with Jordan-Wigner solvable circuits was considered for example in \cite{viktor-floquet}, a free
fermionic cellular automaton was treated in \cite{prosen-150}, and free fermions also appear at the integrable point of
the dual unitary kicked Ising model
\cite{kicked-ising-eredeti,kicked-Ising-space-time-duality-1,dual-unitary-1,dual-unitary-2}.

Regarding hidden free fermions, most of the works focused only the spectrum of the models. The recent work
\cite{sajat-ffd-corr} also studied 
real-time evolution in the model of Fendley, and it demonstrated that exact and explicit formulas can be found for
certain real time correlation functions. The work \cite{sajat-ffd-corr} treated continuous time evolution generated by
the Hamiltonian of \cite{fendley-fermions-in-disguise}, but it also considered unitary quantum circuits that are compatible
with the previously mentioned Hamiltonian.

In \cite{sajat-ffd-corr} it was crucial that the circuits themselves are  free fermionic,
so that exact results can be derived also for discrete time evolution generated by the circuits.
A naive Trotterization following a simple exponentiation
of a free fermionic Hamiltonian could have been used to derive numerical results, but the key point of
\cite{sajat-ffd-corr} was to have exact results for all possible Trotterization steps.
This problem is analogous to the idea of
  ``integrable Trotterization'' \cite{integrable-trotterization} in the sense that there is a distinguished property of
  the Hamiltonian (free 
  fermionic spectrum and/or the integrability), which is broken by many naive Trotterization schemes, but a carefully chosen
  scheme can preserve it without discretization errors.

The work \cite{sajat-ffd-corr} only treated the 4-fermion model of Fendley, and in that model it only considered
one special circuit 
geometry (ordering of local gates). The depth of that circuit was linear in the 
system size. For practical purposes it is desirable to construct circuits with constant depth. This gives a strong motivation
to study alternative circuit geometries. Also, it is desirable to investigate other models, to show that 
exact computations like those presented in \cite{sajat-ffd-corr} are not limited to the 4-fermion model.
  
As a second main motivation we mention the broad field of quantum information science.
Circuits with gates that are bi-linear in fermions are known as matchgate circuits, which are important examples for
classically simulable quantum computation \cite{matchgates1,matchgates2,matchgates3}. Unitary quantum gates based on free
fermions have found applications even in the tensor network discretizations of the AdS/CFT duality
\cite{holocode-review-jahn-eisert}. Hamiltonians with hidden free fermions have found applications in error correcting
codes \cite{free-fermion-subsystem-codes}.
These successes suggest, that quantum circuits with free fermions in disguise could be useful also for quantum information
science.

The recent work \cite{sajat-ffd-corr} demonstrated that the previously mentioned quantum circuits can be
implemented in present day quantum computers, and that a subset of the  dynamical processes they generate are classically
simulable. This raises the question: what are all the possible quantum circuits that are solvable with hidden free
fermions? What is the widest class of classically simulable quantum processes in these models?

\bigskip

In this paper we make the first steps towards answering the questions above. We
investigate unitary quantum circuits with local gates in models with hidden free fermions. We treat the
original model of Fendley \cite{fendley-fermions-in-disguise} but also the generalized model discovered in
\cite{sajat-FP-model}. We derive new free fermionic quantum circuits for these two models.
For the model of Fendley we construct finite depth circuits, which take the form of a special
type of  brickwork circuit. This circuit does not commute with any previously published Hamiltonian, and its hidden free
fermionic structures are currently unknown; their existence is only confirmed by numerical data from exact
diagonalization. 
In contrast, for the model of \cite{sajat-FP-model} we build a linear depth circuit that is compatible
with the Hamiltonian. In this case the hidden free fermions were found in \cite{sajat-FP-model}, and the technical
contribution of this work is to factorize the transfer matrix into local unitary gates.

The structure of the paper is as follows. In Section \ref{sec:problem} we discuss the free fermionic properties in
general, and we formulate our main questions. In Section  \ref{sec:JW} we consider Jordan-Wigner diagonalizable
models; here we just collect well known results.  In Section \ref{sec:ffd} we introduce the
FFD algebra (Free Fermions in Disguise), on which the 4-fermion model of Fendley was built. We discuss basic properties
of this algebra, together with new results for the ``minimal'' representations. We also collect the key results of
\paulcite about the solutions of these models. In Section \ref{sec:ffdc1} we present a free fermionic quantum circuit
within the FFD 
algebra, derived from the results of \cite{fendley-fermions-in-disguise}. In Section \ref{sec:small} we investigate
 the free fermionic circuits for small system sizes. In Section \ref{sec:numerics} we
present new free fermionic circuits with different geometries; here the free fermionic nature is supported by numerical
data.
{Finally, in Section \ref{sec:FP} we treat the generalization of the FFD algebra~\cite{sajat-FP-model} and the free-fermionic model built on that algebra.
We factorize the transfer matrix obtained in \cite{sajat-FP-model}, i.e., we represent the transfer matrix as a product of local gates in the generalized FFD algebra.
Then we compute a unitary quantum circuit also for this model.}
Finally, our conclusions are presented in Section \ref{sec:concl}, and a few details
about the construction of the free fermionic operators are provided in Appendix \ref{sec:aferm}.
{
Appendix~\ref{sec:appnum} contains the numerical data supporting the conjecture in Section~\ref{sec:numerics}. Appendix~\ref{app:FPproof} provides a detailed proof of the factorization of the transfer matrix for the generalized model~\cite{sajat-FP-model}.
}

{For the convenience of the reader we summarize here the most important new results of this work: 
\begin{itemize}
\item In Section \ref{sec:ffdc1} we derive a unitary staircase circuit for the FFD model, building on the results of
  \cite{fendley-fermions-in-disguise}. 
\item In Section \ref{sec:numerics} we obtain new free fermionic quantum circuits in the FFD algebra, including a circuit
  with finite depth.
\item In Section \ref{sec:FP}, we derived the factorized form of the transfer matrix obtained in \cite{sajat-FP-model}, from which we obtain a free fermionic staircase circuit in the generalized FFD algebra.
\end{itemize}
}

\section{The problem}

\label{sec:problem}

We are dealing with quantum spin systems made out of qubits, and we consider one dimensional systems (spin
chains). In all of our examples $L$ will denote the number of sites of the model, which means that the Hilbert space is
$\HH= (\complex^2)^{\otimes L}$. Pauli operators acting on site $j$ will be denoted as $X_j, Y_j, Z_j$. 

In the following we first discuss the definitions and the basic properties of the free fermionic models, both in the
Hermitian and the unitary (quantum circuit) situation. Afterwards in Subsection \ref{sec:posing} we pose our main
problem. 

\subsection{Free fermionic Hamiltonians}

We consider Hermitian Hamiltonians $H$ acting on the Hilbert space. We are interested in those models which have a
free fermionic structure, either obvious or hidden. The free fermions can be introduced as creation/annihilation
operators for the free eigenmodes, satisfying the standard fermionic exchange relations. However, if these structures are
hidden, one needs to 
look for more accessible properties of the models. Therefore, we approach the problem through the spectrum of the
respective models.

We say that a certain Hamiltonian has a free
fermionic spectrum in a finite volume $L$, if the following are satisfied:
\begin{enumerate}
\item 
There are $n$ real numbers $\eps_k$, $k=1,\dots,n$ such that the
distinct eigenvalues $\lambda$ are of the form
\begin{equation}
  \lambda=\sum_{k=1}^n \pm \eps_k.
\end{equation}
\item For generic values of $\eps_k$ every such eigenvalue has the same degeneracy in a chosen volume $L$.
\end{enumerate}
The first criterion implies that the signs can be chosen independently, and the choices do not affect the values of the
$\eps_k$. If the $\eps_k$ are generic numbers, then this gives $2^n$ distinct eigenvalues; the number $n$ is  seen
as the number of fermionic eigenmodes in the given volume $L$.

The second criterion might not be obvious at first sight. After all, in many familiar models the typical degeneracies
are 1, 2 and 4 (often related to spin and space reflection symmetries), and the $\eps_k$ might satisfy some additional
relations, eventually leading to further degeneracies which often depend on the energy of the level. However, if the
geometric symmetries are broken, then typically there is no remaining special relation between the $\eps_k$, and the
degeneracies become identical for each level. In standard examples this remaining degeneracy is typically just one (or a
small constant which is volume 
independent). In contrast, in the recently discovered models there is a large uniform degeneracy which grows exponentially
with the volume \cite{fendley-fermions-in-disguise,sajat-FP-model}.

These two criteria guarantee the existence of fermionic creation/annihilation operators,
satisfying the canonical commutation relations. If every eigenvalue is simply degenerate, then the fermionic operators are
given simply by the ladder operations, with the addition of some properly chosen sign factors. The situation is more
complicated  in the presence of degeneracies, because then there is a large gauge freedom for the choice of the
fermionic operators. This is discussed in detail in Appendix \ref{sec:aferm}.

Standard examples for free fermionic models are the spin chains solvable by the Jordan-Wigner
transformation, in the case of open boundary conditions. Additional examples are the recently discovered models
\cite{fendley-fermions-in-disguise,sajat-FP-model} together with the earlier models  \cite{cooper-anyon,gyuri-susy-1}
and many other examples introduced in  \cite{fermions-behind-the-disguise,unified-graph-th}.

Our definition of a free fermionic spectrum is strict: it excludes models which split into different sectors, such that
the spectrum is free within each sector individually. Examples for this include the Jordan-Wigner solvable spin chains with
periodic boundary conditions (which split into two sectors corresponding to the fermionic parity) and Kitaev's honeycomb
model (which splits into exponentially many symmetry sectors)  \cite{kitaev-honeycomb}. We chose the restricted
definition, because we are mainly interested in the circuits built for the models of the papers
\cite{fendley-fermions-in-disguise,sajat-FP-model,cooper-anyon,gyuri-susy-1}, which fall into this restricted category.

All of our models are defined with open boundary conditions.
As mentioned above, the Jordan-Wigner solvable models with periodic boundary conditions belong to a different, wider
class of free models. The situation is worse in the recently
discovered models of \cite{fendley-fermions-in-disguise,sajat-FP-model}:
 it is known that the models
are integrable in the periodic case, but the free fermionic nature appears completely broken. At present it is
not known how to compute the spectrum with periodic boundary conditions.

To summarize the situation, we apply the  strict definition of a free fermionic spectrum, and all our models will be
defined with open boundary conditions. 

\subsection{Free fermionic unitary operators}

A natural question is how to extend the free fermionic structures to unitary operators. Eventually we will be
interested in unitary operators that are constructed out of local quantum gates acting on the Hilbert space of the spin
chain. This is motivated by real world experiments with quantum computers.
However, for the moment let us give a more general definition for free fermionic unitary operators.

We say that the unitary operator $\VV$ acting on the Hilbert space $\HH$ has a free fermionic spectrum, if
\begin{enumerate}
\item There are $n$ real numbers $\eps_k$, $k=1,\dots,n$ such that the 
distinct eigenvalues $\lambda$ are of the form 
\begin{equation}
  \label{fermioniclambda}
  \lambda=\prod_{j=1}^n e^{\pm i\eps_k}.
\end{equation}
\item If the $\eps_k$ are generic numbers, then every eigenvalue has the same degeneracy.
\end{enumerate}
Once again, the signs can be chosen independently, and these choices do not affect the values of the $\eps_k$. Note that
the $\eps_k$ are defined only up to $2\pi$, therefore they are often called quasi-energies. 

The free fermionic property can be formulated also on the level of the operators. The eigenvalues
\eqref{fermioniclambda} imply the existence of operators $\ZZ_k$ with the properties that
\begin{equation}
  (\ZZ_k)^2=1,\qquad [\ZZ_k,\ZZ_j]=0,
\end{equation}
and the factorization into ``separated variables''
\begin{equation}
  \VV=\prod_{j=1}^n  e^{i\eps_k \ZZ_k}=\prod_{j=1}^n  (\cos(\eps_k)+i\sin(\eps_k) \ZZ_k).
\end{equation}
Applying the arguments of Appendix \ref{sec:aferm} we see that the operators $\ZZ_k$ can be identified with
\begin{equation}
  \ZZ_k=[\Psi_k,\Psi_{-k}],
\end{equation}
with the $\Psi_k$ and $\Psi_{-k}=(\Psi_k)^\dagger$ being canonical free fermionic operators. The fermionic operators are
unique if there are no degeneracies in the spectrum, otherwise there is a gauge freedom associated with them; this is
discussed in detail in Appendix \ref{sec:aferm}.

A formal connection between the Hamiltonians and the unitary operators is given by the exponentiation. If $H$ is a
Hermitian operator with a free fermionic spectrum, then clearly
\begin{equation}
  \label{UH}
  \VV=e^{-i Ht}
\end{equation}
is a unitary operator with free fermionic spectrum for any real $t$. However, even if $H$ is sufficiently local in
space, its exponentiated version includes very non-local terms, and there is no way to reproduce $\VV$ with local
quantum operations.

Below we are interested in a ``free fermionic Trotterization'' of a free fermionic Hamiltonian. The idea is to start
with a free fermionic Hamiltonian $H$ and to construct an approximation of \eqref{UH} using only sufficiently local
unitary quantum gates. This can be understood as a discretization of the continuous time parameter $t$. Most
discretization schemes would break the free fermionic spectrum, and we are interested in those ones that
keep the free fermionic nature without discretization errors.

Below we formulate the problem in a somewhat more general way, without referring to the goal of discretization.

\subsection{Posing the problem}

\label{sec:posing}

Let us assume we have a family of Hamiltonians acting on $\HH$, expressed as
\begin{equation}
  \label{Hdef}
  H=\sum_{j=1}^M \alpha_j h_j,\qquad \alpha_j\in\valos
\end{equation}
where the $h_j$ are Pauli strings (products of Pauli operators acting on selected spins). In most of the examples the
$h_j$ are well localized operators, which act on a few neighbouring sites. The number $M$ stands for the number of the
``fundamental operators'' of the Hamiltonian and  typically it is related to the number of sites $L$.

For each $h_j$ we construct a unitary gate with an angle $\varphi_j\in\valos$:
\begin{equation}
  \label{udef}
  u_j({\varphi_j})=e^{i\varphi_j h_j}=\cos(\varphi_j)+i\sin(\varphi_j) h_j.
\end{equation}
{We note that the dependence of $\varphi_j$ on $u_j(\varphi_j)$ is often omitted in the following for simplicity.}
We then construct various quantum circuits as products of these local gates.  Let
$\SSS=\{s_1,s_2,\dots,s_N\}$ be a 
string of indices, with length $N$, such that $s_k\in [1,M]$. Repetitions of indices are allowed, and generally we
expect that $N$ will grow linearly with $M$. Then we take the product
\begin{equation}
  \label{productdef}
  \VV=u_{s_N}\dots u_{s_2}u_{s_1}.
\end{equation}
We call the sequence of operators in the product the ``geometry of the circuit''. For a graphical representation of such
circuits see Fig. \ref{fig:cggt}.

A few fundamental
properties of the spectrum of such $\VV$ operators are established later in Section \ref{sec:prop1}.

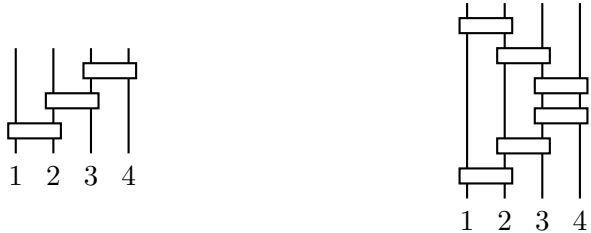
\begin{figure}[t]
  \centering
  \begin{tikzpicture}
    \draw [thick] (0,0) -- (0,1.4);
    \draw [thick] (0.5,0) -- (0.5,1.4);
    \draw [thick] (1,0) -- (1,1.4);
         \draw [thick] (1.5,0) -- (1.5,1.4);
       \draw [thick, fill=white] (-0.1,0.2)  rectangle (0.6,0.4);
       \draw [thick, fill=white] (0.4,0.6)  rectangle (1.1,0.8);
       \draw [thick, fill=white] (0.9,1)  rectangle (1.6,1.2);
       \node at (0,-0.3) {$1$};
       \node at (0.5,-0.3) {$2$};
       \node at (1,-0.3) {$3$};
            \node at (1.5,-0.3) {$4$};

            \begin{scope}[xshift=6cm,yshift=-0.6cm]
    \draw [thick] (0,0) -- (0,2.6);
    \draw [thick] (0.5,0) -- (0.5,2.6);
    \draw [thick] (1,0) -- (1,2.6);
         \draw [thick] (1.5,0) -- (1.5,2.6);
       \draw [thick, fill=white] (-0.1,0.2)  rectangle (0.6,0.4);
       \draw [thick, fill=white] (0.4,0.6)  rectangle (1.1,0.8);
       \draw [thick, fill=white] (0.9,1)  rectangle (1.6,1.2);
       \draw [thick, fill=white] (0.9,1.4)  rectangle (1.6,1.6);
       \draw [thick, fill=white] (0.4,1.8)  rectangle (1.1,2);
       \draw [thick, fill=white] (-0.1,2.2)  rectangle (0.6,2.4);

       \node at (0,-0.3) {$1$};
       \node at (0.5,-0.3) {$2$};
       \node at (1,-0.3) {$3$};
       \node at (1.5,-0.3) {$4$};
       \end{scope}
   \end{tikzpicture}
  \caption{Pictorial representation for two simple quantum circuits acting on $L=4$ spins. The vertical lines stand for
    the spins, and it 
    is understood that  time flows upwards. Here each gate $u_j$ acts on two neighbouring spins at sites $j$ and $j+1$.
    The circuit on the left is $\VV=u_3u_2u_1$, while the circuit on the
    right is $\VV=u_1u_2u_3u_3u_2u_1$.}
  \label{fig:cggt}
\end{figure}

It is important that if a certain gate $u_j$ is applied multiple times,
then it should always have the same angle $\varphi_j$. This is the only fine tuning that we require in the problem, and
later we will see that in most cases this fine tuning is indeed essential.

\bigskip

In the original model of Fendley and in the models treated with the graph theoretical framework of
\cite{fermions-behind-the-disguise,unified-graph-th} the Hamiltonian has a free fermionic spectrum for
any choice of the coefficients $\alpha_j$. In such models it is natural to ask the 
following question.

\bigskip

{\bf Question 1a:} For which choice of indices $\SSS$ is the resulting unitary operator $\VV$ free fermionic, {\it for any choice
  of the $\varphi_j$ parameters}? In other words: which geometry of a quantum circuit will keep the free fermionic
nature of the original Hamiltonian?

\bigskip

There are other types of models, which have hidden free fermions only in those cases, when a certain degree
of  fine tuning is applied to the couplings $\alpha_j$ in the Hamiltonian. An example is the model discovered in \cite{sajat-FP-model}.
In such cases the relevant question is:

\bigskip

{\bf Question 1b:} For which choice of indices $\SSS$ and for which angles $\varphi_j$ is 
the resulting unitary operator $\VV$ free fermionic?

\bigskip

These questions can be investigated for concrete examples with small $M, N$, but eventually we would like to approach
large systems as well. Therefore, we are especially interested in those geometries which can be defined for a set of
increasing values of $M$ and $N$. In certain cases we obtain the so-called ``brickwork circuits'', which can be seen as
a particular prescription for Floquet time evolution. The idea is to compose the circuit out of finite number of ``time
steps'', such that in each time step we apply a sequence of gates, all of which commute with each other. Gates that
appear in different time steps are generally not commuting, leading to the actual dynamics in the problem. 

If the $h_j$ operators involve only nearest neighbour interactions (acting on sites $j$ and $j+1$) , then the the
standard brickwork circuit in a finite volume $L=2k$ is constructed as
\begin{equation}
  \VV=(u_2u_4\dots u_{L-2})(u_1u_3\dots u_{L-1}).
\end{equation}
Here the parentheses are added only for better readability. It follows from the nearest neighbour property that the $u_k$
operators commute with each other within each parenthesis, but generally they are not commuting. For a graphical
representation of this circuit see Fig. \ref{fig:cbw}. Brickwork circuits for more general geometries will be presented
later in Section \ref{sec:numerics}; an example can be seen in Fig.  \ref{fig:cbw3}.

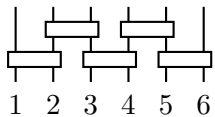
\begin{figure}[t]
  \centering
  \begin{tikzpicture}
    \draw [thick] (0,0) -- (0,1);
    \draw [thick] (0.5,0) -- (0.5,1);
    \draw [thick] (1,0) -- (1,1);
    \draw [thick] (1.5,0) -- (1.5,1);
    \draw [thick] (2,0) -- (2,1);
          \draw [thick] (2.5,0) -- (2.5,1);
          \draw [thick, fill=white] (-0.1,0.2)  rectangle (0.6,0.4);
          \draw [thick, fill=white] (0.9,0.2)  rectangle (1.6,0.4);
          \draw [thick, fill=white] (1.9,0.2)  rectangle (2.6,0.4);
          
          \draw [thick, fill=white] (0.4,0.6)  rectangle (1.1,0.8);
            \draw [thick, fill=white] (1.4,0.6)  rectangle (2.1,0.8);
       \node at (0,-0.3) {$1$};
       \node at (0.5,-0.3) {$2$};
       \node at (1,-0.3) {$3$};
       \node at (1.5,-0.3) {$4$};
       \node at (2,-0.3) {$5$};
           \node at (2.5,-0.3) {$6$};
   \end{tikzpicture}
  \caption{Pictorial representation for a ``brickwork circuit'' acting on $L=6$ spins, with open boundary
    conditions. Here every gate $u_j$ acts on two neighbouring sites and the circuit is $\VV=u_2u_4u_1u_3u_5$.}
  \label{fig:cbw}
\end{figure}

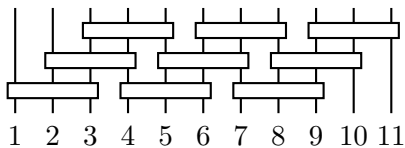
\begin{figure}[t]
  \centering
  \begin{tikzpicture}
    \draw [thick] (0,0) -- (0,1.4);
    \draw [thick] (0.5,0) -- (0.5,1.4);
    \draw [thick] (1,0) -- (1,1.4);
    \draw [thick] (1.5,0) -- (1.5,1.4);
    \draw [thick] (2,0) -- (2,1.4);
    \draw [thick] (2.5,0) -- (2.5,1.4);
    \draw [thick] (3,0) -- (3,1.4);
    \draw [thick] (3.5,0) -- (3.5,1.4);
    \draw [thick] (4,0) -- (4,1.4);
    \draw [thick] (4.5,0) -- (4.5,1.4);
    \draw [thick] (5,0) -- (5,1.4);
        
          \draw [thick, fill=white] (-0.1,0.2)  rectangle (1.1,0.4);
          \draw [thick, fill=white] (1.4,0.2)  rectangle (2.6,0.4);
          \draw [thick, fill=white] (2.9,0.2)  rectangle (4.1,0.4);
          
          \draw [thick, fill=white] (0.4,0.6)  rectangle (1.6,0.8);
          \draw [thick, fill=white] (1.9,0.6)  rectangle (3.1,0.8);
          \draw [thick, fill=white] (3.4,0.6)  rectangle (4.6,0.8);

                   \draw [thick, fill=white] (0.9,1)  rectangle (2.1,1.2);
                   \draw [thick, fill=white] (2.4,1)  rectangle (3.6,1.2);
                     \draw [thick, fill=white] (3.9,1)  rectangle (5.1,1.2);
       \node at (0,-0.3) {$1$};
       \node at (0.5,-0.3) {$2$};
       \node at (1,-0.3) {$3$};
       \node at (1.5,-0.3) {$4$};
       \node at (2,-0.3) {$5$};
       \node at (2.5,-0.3) {$6$};
       \node at (3,-0.3) {$7$};
       \node at (3.5,-0.3) {$8$};
       \node at (4,-0.3) {$9$};
       \node at (4.5,-0.3) {$10$};
        \node at (5,-0.3) {$11$};
   \end{tikzpicture}
  \caption{A generalized ``brickwork circuit'' acting on $L=11$ spins, with open boundary
    conditions. Here the gates $u_j$ act on three neighbouring sites $j, j+1, j+2$, and the circuit is
    $\VV=(u_3u_6u_9)(u_2u_5u_8)(u_1u_4u_7)$. Parentheses are added for better readability, and to signal the groups of
    gates which are applied in a given ``layer'' or ``time step''.}
  \label{fig:cbw3}
\end{figure}

\bigskip

Let us now discuss the algebraic background for the problem.

For a certain $M$ we define the {\it bond algebra} as the
operator algebra generated by the $h_j$, $j=1,\dots,M$. This is a finite dimensional algebra, and it is a sub-algebra of
the full operator algebra of the spin chain of length $L$.

Every unitary operator of the form \eqref{productdef} is a member of the bond algebra.
Once we expand the product in  \eqref{productdef} using \eqref{udef} we typically obtain a high order polynomial in
the $h_k$.

Let us now investigate the limit of small angles $\varphi_k$. In this case the quantum circuit $\VV$ can be regarded as a 
discretization of the unitary operator  
\begin{equation}
  e^{-iH \Delta t},
\end{equation}
where $\Delta t$ is a small number (compared to the energy scales of $H$).
In such a case we can expand the local unitary
operators as
\begin{equation}
  u_k\approx 1+i\varphi_k h_k.
\end{equation}
Let us assume that each index $k$ appears in $\SSS$ the same number of times {(the number of repetitions is given by
$\rho=N/M$),}
and let us substitute
\begin{equation}
  \label{dtlimit}
\varphi_j=-\alpha_j \frac{\Delta t}{\rho},\qquad \Delta t\to 0.
\end{equation}
 Then we get to leading order in $\Delta t$:
\begin{equation}
  \VV\approx 1-iH \Delta t,\qquad H=\sum_k  \alpha_k h_k.
\end{equation}
This means that small angles lead to a spectrum which is  {\it approximately free fermionic for any choice of
  the geometry}. In this limit
the unitary operator commutes with $H$, which is linear in the $h_k$. However, discretization effects (higher order
terms in the angles) would very often destroy the free fermionic nature.

If we expand the product \eqref{productdef} in the bond algebra, we can always take the formal logarithm of the
operator $\VV$. This will result in a formal power series in the angles $\varphi_k$, with some
radius of convergence. If a certain choice of angles is within the radius of convergence, then we obtain a well defined
Hermitian operator $H^*$ such that
\begin{equation}
  \VV=e^{i H^*}.
\end{equation}
In certain concrete cases $H^*$ can be found explicitly as a function of the angles $\varphi_k$. Then a free fermionic
circuit will lead to a free fermionic Hermitian $H^*$.  Generally $H^*$ will involve terms that are
increasingly non-local (high order in the $h_k$), and in the general case finding $H^*$ is not  easier
than actually computing the spectrum of the unitary circuit.

Experience with the free fermionic Hamiltonians and with  integrable quantum circuits in general suggest
that every free fermionic circuit
will have a set of commuting local charges. It is expected that
these charges will be low order polynomials in the bond algebra. This leads us to the following question:

\bigskip

{\bf Question 2:} If a certain circuit is free fermionic, are there Hermitian operators
commuting with $\VV$, which are low order polynomials in the $h_k$? What is the lowest order?

\bigskip

Below we will see that there exist free fermionic unitary circuits for which the lowest order commuting charge is
actually a Hamiltonian of the form \eqref{Hdef}.
In such a case there is a direct mapping between the set of angles $\{\varphi_j\}$ and the parameters $\{\alpha_j\}$ of the
Hamiltonian, {which reduces to \eqref{dtlimit} in the infinitesimal angle limit}. However, generally it is not possible to find a Hamiltonian of the form \eqref{Hdef} commuting with $\VV$,
and the ``smallest'' commuting charge emerges at a slightly higher order.

\subsection{Basic properties of the circuits}

\label{sec:prop1}

Here we establish a few simple properties of the circuits $\VV$ and their spectrum.

\medskip

{\bf Property 1.} The spectrum of $\VV$ is invariant with respect to a cyclic shift of the operator product \eqref{productdef},
which can also be described by a cyclic shift of the index set $\SSS$. 

\medskip

{\bf Property 2.} If for some $k$ we have the commutativity $[h_{s_{k+1}},h_{s_k}]=0$, then this  extends to the
commutativity $[u_{s_{k+1}},u_{s_k}]=0$, which implies that the order of the products (or the index set $\SSS$) can be
rearranged accordingly.

\medskip

{\bf Property 3.} A directly repeated application of the same gate is equivalent to a single application, with regards
to the free fermionic property. For example, if there is a repeated action of the gate $u_1$ in the operator product,
then we can replace
{\begin{align}
  u_1{(\varphi_1)}u_1{(\varphi_1)}
  &=
  e^{i\varphi_1 h_1}e^{i\varphi_1 h_1}=e^{i2\varphi_1 h_1}
  \nonumber\\
  &=
{u_1(2\varphi_1)}.
\end{align}}
This corresponds to the action of a single $u_1$, but now with the doubled angle.

\medskip

Sometimes the combination of these properties will lead to simplifications of the problem. {For example,
consider the free fermionic nature of the the circuit
\begin{equation}
 {u_1(\varphi_1)u_2(\varphi_2)u_3(\varphi_3)u_2(\varphi_2)u_1(\varphi_1)}
\end{equation}
for some triplet of operators $h_1, h_2, h_3$, such that $[h_3,h_2]=0$. Then by Property 1. and 2. we
are free to investigate the circuit {$u_3(\varphi_3)u_2(\varphi_2)u_2(\varphi_2)u_1(\varphi_1)u_1(\varphi_1)$}, and by Property 3. we can simplify it and investigate {$u_3(\varphi_3)u_2(2\varphi_2)u_1(2\varphi_1) = u_3(\varphi^\prime_3)u_2(\varphi^\prime_2)u_1(\varphi^\prime_1)$} instead. }

For quantum circuits with nearest neighbour two-site gates the spectral equivalences under the re-ordering of the gates
were investigated recently in \cite{znidaric-classification}. For gates acting on more than two sites we are not aware
of such a study. 

\section{Jordan-Wigner diagonalizable models}

\label{sec:JW}

Now we consider models that can be diagonalized with a Jordan-Wigner transformation. We consider an abstract formulation
of the transformation, which does not depend on the concrete representation. The class of models for which such a
transformation exists was determined in the recent works \cite{chapman-jw,japan-free-fermion-JW}.

Let us assume that there are Majorana fermionic operators $\chi_m$, $m=1,\dots,N_\chi$, such that
\begin{equation}
  \label{chidef}
  \chi_m^2=1,\qquad \{\chi_m,\chi_n\}=\delta_{m,n}
  \,,
\end{equation}
and every $h_j$ is bilinear in the $\chi_m$.

\begin{thm}
  If every $h_j$ is bilinear in the $\chi_m$, then the Hamiltonians \eqref{Hdef} are free fermionic. 
\end{thm}

This follows simply from the commutation relations of fermionic bilinears. For a concise treatment see for example
\cite{free-parafermion}. 

For the Jordan-Wigner diagonalizable models a simple theorem holds also for the quantum circuits:

\begin{thm}
   If every $h_j$ is bilinear in the $\chi_m$, then every quantum circuit of the form \eqref{productdef} is free fermionic.
\end{thm}
\begin{proof}
 This can be proven by induction over the length of the operator product. We assume that up to a certain length $\ell$ the
product can be written as
\begin{equation}
  e^{i O_\ell},
\end{equation}
where $O_\ell$ is bilinear in the $\chi_j$. This is true for the initial point of the induction with a single $u_j$.
Then we evaluate the next product, and we compute
\begin{equation}
O_{\ell+1}=\log(u_{\ell+1}e^{iO_{\ell}})=\log(e^{i \varphi_{\ell+1} h_{\ell+1}} e^{iO_\ell})
\end{equation}
using the Baker-Campbell-Hausdorff formula. We will end up with an infinite sum of nested commutators. However, at every
step we are computing a commutator of fermionic bilinears, therefore  every intermediate step results in terms bilinear
in fermions. This implies that the resulting operator $O_{\ell+1}$ will be also bilinear in the Majorana fermions.
\end{proof}

\subsection{The quantum Ising chain}

\label{sec:ising}

Let us consider a specific example: the algebra of the quantum Ising chain.

Let $h_k$ with $k=1,\dots,M$ be Hermitian operators satisfying the relations
\begin{equation}
  \label{isingalg}
  \begin{split}
    h_k^2&=1,\qquad k=1,\dots,M\\
    \{h_k,h_{k+1}\}&=0,\qquad k=1,\dots,M-1\\
    [h_k,h_\ell]&=0,\qquad |k-\ell|>1.
     \end{split}
\end{equation}

A representation of this algebra is given by the quantum Ising chain. For example, for odd $M$ the standard
representation is on a chain of length $L=(M+1)/2$
\begin{equation}
  \label{isingrep}
  \begin{split}
    h_{2k-1}&=Z_{k},\quad  k=1,\dots,L\\
    h_{2k}&=X_{k}X_{k+1},\quad k=1,\dots,L-1.
     \end{split}
\end{equation}
An alternative representation can be given in a homogeneous way on a spin chain of length $L=M+1$ by
\begin{equation}
  h_k=X_kY_{k+1},\qquad k=1,\dots,N.
\end{equation}

It is well known that the resulting spin chain Hamiltonians can be solved by Majorana fermions, that arise from the
standard Jordan-Wigner transformation. Introducing the
operators $\chi_m$ we write
\begin{equation}
  h_m=i\chi_m\chi_{m+1},\qquad m=1,\dots,M.
\end{equation}
We see that the defining relations are satisfied. The representation of the Majorana fermions in terms of the spin
operators might depend on the representation of the $h_j$. In the example of \eqref{isingrep} we can choose
\begin{equation}
  \chi_{2k-1}=\left(\prod_{\ell=1}^{k-1} Z_\ell\right)Y_{k},\qquad \chi_{2k}= \left(\prod_{\ell=1}^{k-1} Z_\ell\right)X_{k}.
\end{equation}

Let us now build a brickwork circuit using this algebra. For example we choose $M=2L-1$ with some $L$ and define
\begin{equation}
  \label{ising1}
 \VV=(u_M\dots u_3u_1)(u_{M-1}\dots u_4u_2).
\end{equation}
The operator products commute in each parenthesis. Performing a homogeneous rotation in the concrete representation
\eqref{isingrep} we obtain the Floquet 
time evolution operator of the ``kicked Ising model'' at the free fermionic point
\cite{kicked-ising-eredeti,kicked-Ising-space-time-duality-1,dual-unitary-1,dual-unitary-2}:
\begin{equation}
  \VV=\exp\left(i\sum_{k=1}^L c_k X_k\right) \cdot \exp\left(i\sum_{k=1}^{L-1}  d_k Z_{k}Z_{k+1}\right)
\end{equation}
with the identification of the angles $c_k=\varphi_{2k-1}$, $d_k=\varphi_{2k}$. This operator can be seen as a complex valued
and inhomogeneous generalization of the transfer matrix of the two dimensional Ising model, solved by free fermions in
the historic work 
\cite{Schulz-Mattis-Lieb}. 

It is known that this is a free fermionic circuit. Here we do not discuss its explicit solution, instead we focus on the
conserved charges of the circuit. Direct computation using the algebra \eqref{isingalg} shows that the operator $\VV$
commutes with the extensive Hermitian charge
\begin{equation}
  \label{Qising}
  \begin{split}
    Q=&\sum_{j=1}^M  \sin(2\varphi_{j})(\cos(2\varphi_{j-1})+\cos(2\varphi_{j+1})) h_j+\sum_{j=1}^{M-1}  (-1)^{j+1} i \sin(2\varphi_j)\sin(2\varphi_{j+1}) h_j h_{j+1},
  \end{split}
\end{equation}
where it is understood that $\varphi_0=\varphi_{M+1}=0$. Note that this operator has quadratic terms in the $h_j$, but
they are still bilinear in the fermions, because $h_jh_{j+1}=\chi_j\chi_{j+2}$. 

In this example there is no specialization of the Hamiltonian \eqref{Hdef} which would commute with
the circuit $\VV$, but the low order charge $Q$ does commute. This pattern will be observed also in the circuits with
hidden free fermions: if the geometry is regular, then 
there always exist low order commuting charges. However, their structure becomes more and more complicated as the layers
(or time steps) of the brickwork circuits are increased. This is discussed in Subsection \ref{sec:commc}.

\section{Free fermions in disguise}

\label{sec:ffd}

In this Section we consider the model introduced in \cite{fendley-fermions-in-disguise}. This model is built on the
abstract algebra defined below, which we call the FFD algebra (Free Fermions in Disguise).
 A generalization of this
model, which is built on a modified algebra was later discovered in \cite{sajat-FP-model}. We treat 
 that model in Section \ref{sec:FP}. 

It was proven in \cite{fendley-fermions-in-disguise} that the FFD algebra guarantees the free fermionic nature of the
family of Hamiltonians \eqref{Hdef}. In Subsection \ref{sec:reps} we discuss the 
representations of this algebra, in Subsection \ref{sec:ffdtm} we summarize the main results of
\cite{fendley-fermions-in-disguise} and in Section \ref{sec:ffdc1} we
show that the constructions in \cite{fendley-fermions-in-disguise} directly lead to unitary quantum circuits with
free fermionic spectrum. 

The FFD algebra is generated by the Hermitian operators $h_j$, $j=1,\dots,M$ with the property $h_j^2=1$, satisfying the
following commutation relations:
\begin{equation}
  \label{ffdalgebra}
  \begin{split}
    \{h_j,h_{j+1}\}&=0,\quad j=1,\dots,M-1\\
    \{h_j,h_{j+2}\}&=0, \quad j=1,\dots,M-2 \\
       [h_j,h_{k}]&=0, \quad \text{if } |j-k|>2.
  \end{split}
\end{equation}
The FFD algebra with $M$ generators is $2^M$ dimensional over the complex numbers, because every operator product
\begin{equation}
  h_{j_1}h_{j_2}\cdots h_{j_n}
\end{equation}
can be brought to a ``normal ordered'' form, where the indices of the operators become ordered. For concreteness we choose
the convention that a normal ordered product satisfies $j_1<j_2<\dots<j_n$. Then it follows from the normal ordering
(and the absence of additional algebraic relations), that the operator algebra is 
$2^M$ dimensional, with a basis given by
\begin{equation}
\prod_{j=1}^M (h_{j})^{n_j},\qquad n_j=0, 1,
\end{equation}
with the ordering of the product as specified above.

The algebraic relations \eqref{ffdalgebra} might seem surprising or exotic. However, in the next Subsection we show that
for certain 
values of $M$ the FFD algebra is actually identical to the full operator algebra of a spin chain of length
$L=M/2$. Therefore, the defining relations can be seen as a very special choice for the generating elements of the
standard operator algebra.

For later use let us introduce the linear operation $[.]^T$ on the FFD algebra via its action on the basis elements:
\begin{equation}
  \label{transposedef}
  (h_{j_1}h_{j_2}\cdots h_{j_n})^T=h_{j_n}\cdots h_{j_2}h_{j_1}.
\end{equation}
It follows from the commutation relations, that on the basis elements this operation is always equal
to multiplication by 
$\pm 1$. Then we extend its action to the linear combinations of the basis elements. It follows that for any two
operators $A$ and $B$ in the FFD algebra
\begin{equation}
  (AB)^T=B^TA^T
\end{equation}
The operation is very similar to
the usual transpose, and if we are dealing with a representation where each $h_j$ is a symmetric matrix, then the
operation $[.]^T$ is actually identical to the usual transpose. However, there are representations where some of the
$h_j$ are not symmetric, and in those cases we do not get the usual transpose. It is important that $[.]^T$ is not
equivalent to the usual adjoint either: taking the adjoint of an operator is an anti-linear operation, whereas  $[.]^T$
is linear.

Let us also discuss the center of the FFD algebra. Direct computations show that there
are central elements for $M=6k+1$, $M=6k+3$, $M=6k+4$ and $M=6k+5$. We list them below.

For every $M=3k+1$ we find the central element
\begin{equation}
   C^{(a)}=h_1h_4h_7\dots h_{M},
\end{equation}
while for $M=6k+5$ there is a central element given by
\begin{equation}
  C^{(b)}=(ih_2h_3h_4)(ih_8h_9h_{10})\dots (ih_{M-3}h_{M-2}h_{M-1}).
\end{equation}
Here the parentheses are added for better readability, and the factors of $i$ are added to ensure the Hermiticity of the product.
The operator $C^{(b)}$ can be used to find a central element also for $M=6k+4$ and $M=6k+3$, by ``forgetting'' about the
first and/or the last generators of the FFD algebra (which are not present in $C^{(b)}$). This means that for $M=6k+4$ we find two
central elements  
\begin{equation}
  \begin{split}
  C^{(c)}&=(ih_1h_2h_3)(ih_7h_8h_{9})\dots (ih_{M-3}h_{M-2}h_{M-1})\\
   C^{(d)}&=(ih_2h_3h_4)(ih_8h_9h_{10})\dots (ih_{M-2}h_{M-1}h_{M}).
\end{split}
\end{equation}
Their product is also central, in fact we obtain $C^{(c)}C^{(d)}=C^{(a)}$ for the given $M=6k+4$. Finally, for $M=6k+3$
we obtain the central element
\begin{equation}
   C^{(e)}=(ih_1h_2h_3)(ih_7h_8h_{9})\dots (ih_{M-2}h_{M-1}h_{M}).
\end{equation}
  
All these central elements are Hermitian, they square to one, therefore their eigenvalues are $\pm 1$. The eigenvalue
problem of any  operator can be investigated separately in the sectors specified by the eigenvalues of these central
elements. 

Note that there is no central element for $M=6k$ and $M=6k+2$. This has implications for the ``minimal''
representations discussed below.

\subsection{Representations}

\label{sec:reps}

Here we discuss representations of the algebra defined by relations \eqref{ffdalgebra} with $M$ generators.

For the
applications of the algebra (and for numerical and/or algebraic computations) it is useful to find representations on
spin chains with the shortest length possible. 
The operator algebra of a spin-1/2 chain with length $L$ is $4^L$
dimensional, whereas the FFD algebra is $2^M$ dimensional.
This implies that the minimal length of a faithful representation (for even $M$) is $L=M/2$. We call such faithful
representations ``minimal''.
We will see that depending on $M$ it is possible to construct such minimal representations. 

However, first we consider the simplest representation, which was the main example in the paper
\cite{fendley-fermions-in-disguise}. This is given by the following operators acting
on a spin chain with $M+2$ sites:
\begin{equation}
  \label{basicrep}
  h_j=X_jX_{j+1}Z_{j+2}.
\end{equation}
The advantage of this representation is that it is homogeneous in the index $j$. However, this comes at a cost: the
embedding Hilbert space is rather large, and the model will have large degeneracies on top of those which come from the
hidden free
fermionic nature. In fact, with this representation it becomes possible to embed the same algebra in one more instance
into the same Hilbert space. 

To this order consider the operators $\tilde h_j=Z_jX_{j+1}X_{j+2}$. It can be seen that they also form a representation
of the same algebra, and we have the commutation relations
\begin{equation}
  [\tilde h_j,h_k]=0.
\end{equation}
Therefore, this representation allows for the embedding of two independent copies of the same algebra into the given
Hilbert space.

The dimension of the embedding space can be reduced by ``cutting off'' Pauli matrices that act on the first and the last
spins. We end up with the representation on $M$ spins given by
\begin{equation}
  \label{basicrep2}
  \begin{split}
    h_1&=X_1Z_2\\
    h_j&=X_{j-1}X_{j}Z_{j+1},\quad j=2,\dots,M-1\\
    h_M&=X_{M-1}X_M.\\
  \end{split}
\end{equation}
A similar representation can be given for the $\tilde h_j$ operators, thereby faithfully embedding two copies of the
same algebra into the full operator algebra of the Hilbert space. Both copies have dimension $2^M$. There is no
additional algebraic relation satisfied by the two representations, therefore we obtain the full operator algebra of the
spin chain of length $L=M$ as the product of two commuting copies of the FFD algebra with $M$ generators. 

For the special case of $M=6k$ an alternative representation was also given in \cite{fendley-fermions-in-disguise}, with
an embedding into a spin chain of length $L=3k+2$. Here the first six generators are represented as \footnote{Here and
  in the following we tabulate the formulas so that the Pauli operators acting on a given site are placed below each
  other. This makes it easy to check the commutation relations. }
\begin{center}
  \label{frep1}
  \begin{tabular}{*{6}{c}}
    $h_1=$   & $X_1$  &&&&  \\
    $h_2=$   & $Z_1$  & & $Z_3$ &&  \\
    $h_3=$   &$Z_1$ & $X_2$& $Y_3$  &&  \\
    $h_4=$   & & & $X_3$  &&  \\
    $h_5=$   & & & $Z_3$ & $Z_4$  &  \\
    $h_6=$   & & $Z_2$ & $Z_3$  & $Y_4$ & $Z_5$,   \\
  \end{tabular}
\end{center}
and the $h_j$ for $j>6$ are obtained by a shift with 3 sites:
\begin{equation}
  h_j=\eta^3(h_{j-6}),
\end{equation}
where $\eta$ is an operator that shifts an operator to the right by one site. For example, the next operators are
\begin{equation}
  h_7=X_4,\qquad h_8=Z_4Z_6.
\end{equation}
This representation is not minimal, but it can be used to construct minimal ones, as we show below\footnote{The minimal
  representations we present are built on earlier results of \paulcitee, nevertheless they are new results of this work.}.

First of all note that this representation can be used for any $M$ (not necessarily a multiple of 6) simply by omitting
a small number of $h$'s 
either from the beginning or the end of the list of operators. Also, the representation can be ``shortened'' for $M=6k$
by chopping off the last Pauli matrix in $h_{6k}$, instead using
\begin{equation}
  h_{6k}=Z_{3k-1}Z_{3k}Y_{3k+1}.
\end{equation}
In this case the spin chain has length $L=3k+1$.

Furthermore, simple  modifications do actually achieve the minimal length in the cases $M=6k$ and $M=6k+2$.
In the case of $M=6k+2$ let us take the operators as defined above for $j=1,\dots,6k-1$ and for the last operator let us
chop off of the last Pauli matrix, so
\begin{equation}
  h_{6k}=Z_{3k-1}Z_{3k}Y_{3k+1}.
\end{equation}
Furthermore let us include two new generators $h_0$ and $h_{6k+1}$ defined as
\begin{equation}
  h_0=Y_1Z_2,\qquad h_{6k+1}=X_{3k+1},
\end{equation}
and let us shift the indices of the $h_j$ operators by one.
In total this gives $M=6k+2$ generators acting on a spin chain with length $L=3k+1$. Direct check shows that there is no
additional relation satisfied by the representants, therefore the representation is faithful.
In the simplest case of $k=1$ this gives the following $M=8$ operators:
\begin{center}
  \begin{tabular}{*{5}{c}}
$h_1=$ &$Y_1$ & $Z_2$ && \\
    $h_2=$   & $X_1$  &&& \\
    $h_3=$   & $Z_1$  & & $Z_3$ &  \\
    $h_4=$   &$Z_1$ & $X_2$& $Y_3$  &  \\
    $h_5=$   & & & $X_3$  & \\
    $h_6=$   & & & $Z_3$ & $Z_4$    \\
    $h_7=$   & & $Z_2$ & $Z_3$  & $Y_4$    \\
    $h_8=$ &&&& $X_4$.
  \end{tabular}
\end{center}

Let us also construct a representation of minimal length for $M=6k$. In this case the idea is to consider the sequence
of operators prescribed by \eqref{frep1}, to ``start'' the sequence from $h_5$, with ``cutting off'' some of the Pauli
matrices. More concretely in the simplest case of $M=6$ we have
\begin{center}
  \begin{tabular}{*{4}{c}}
    $h_1=$  &  $Z_1$  && \\
    $h_2=$   &  $Y_1$ & $Z_2$ &  \\
    $h_3=$   &  $X_1$&&  \\
    $h_4=$   &  $Z_1$ & & $Z_3$    \\
    $h_5=$   &  $Z_1$& $X_2$& $Y_3$  \\
 $h_{6}=$   & && $X_3$.   \\
  \end{tabular}
\end{center}
For $M=6k$ with $k>1$ we continue the representation as
\begin{equation}
  h_7=Z_3Z_4,\qquad h_8=Z_2Z_3Y_4Z_5,
\end{equation}
and afterwards
\begin{equation}
  h_j=\eta^3(h_{j-6}),\ \text{for } 8<j\le M.
\end{equation}
Altogether this gives a minimal representation for $M=6k$.

In the previous Subsection it was established that if $M$ is even, then there are two central elements for $M=6k+4$. This
implies that for $M=6k+4$ we can't find a faithful representation with minimal length $L=M/2$, because the full operator
algebra of the spin chains does not have central elements other than the identity. This is in agreement with the
previous constructions, which provide minimal representations for $M=6k$ and $M=6k+2$, but not for $M=6k+4$. 

For completeness we provide Table \ref{tab:min}, which shows the smallest length $L$ of the embedding  spin chains for
the faithful representations. 

\begin{table}[t]
  \centering
  \begin{tabular}{|c||c|c|c|c|c|c|}
    \hline
$M$    & $6k$   & $6k+1$ & $6k+2$ & $6k+3$ & $6k+4$ & $6k+5$ \\
           \hline 
  $L$  & $3k$ & $3k+1$  &  $3k+1$   & $3k+2$  &   $3k+3$ & $3k+3$ \\
    \hline
 $n_c$ &  0 & 1 & 0 & 1 & 2 & 1   \\
    \hline
  \end{tabular}
  \caption{List of the minimal length $L$ of the spin chains for the minimal representations of the FFD algebra with $M$
    generators, together with the number of algebraically independent central elements $n_c$.  These numbers satisfy the
    relation $2L=M+n_c$. Concrete expressions for the selected central elements are given above in the main text.} 
  \label{tab:min}
\end{table}

\subsection{Solution of the Hamiltonians in the FFD algebra}

\label{sec:ffdtm}

Now we summarize the results of \paulcite about the spectrum of the Hamiltonians in the FFD algebra.
With some slight change of notations as opposed to  \cite{fendley-fermions-in-disguise} we introduce
\begin{equation}
  H=\sum_{j=1}^M H_j,\qquad H_j=\alpha_jh_j,\quad \alpha_j\in\valos.
\end{equation}
It was shown in  \cite{fendley-fermions-in-disguise} that this family of models has a set of commuting higher
charges. The first charge is defined as
\begin{equation}
  \label{Q2def}
  Q^{(2)}=\sum_{m>m'+2} H_m H_{m'},
\end{equation}
which is actually equal to $H^2/2$ plus a constant. Higher charges are defined in a similar way: for $Q^{(s)}$ we sum all 
products of $s$ number of $H_m$  operators such that they all commute with each other. Formally we can write
\begin{equation}
  \label{Qsdef}
  Q^{(s)}=\sum_{m_{r+1}>m_{r}+2} H_{m_1}H_{m_2}\dots H_{m_s}.
\end{equation}

Furthermore we define the transfer matrix with a complex parameter $v$ as
\begin{equation}
  \label{Tdef}
  T_M(v)=\sum_{s=0}^S (-v)^sQ^{(s)},
\end{equation}
where it is understood that $Q^{(0)}=1$ and $Q^{(1)}=H$ and $S \equiv [(M+2)/3]$ with $[x]$ the integer part of $[x]$.

It can be seen from the construction of the charges, that the transfer matrix satisfies the recursion relation
\begin{equation}
  T_M(v)=T_{M-1}(v)-vH_M T_{M-3}(v),
\end{equation}
with the initial conditions $T_M=1$ for $M\le 0$.

Let us now also define a polynomial $P_M(v^2)$ using the recursion relation
\begin{equation}
  \label{PMrec}
  P_M(v^2)=P_{M-1}(v^2)-v^2\alpha_M^2 P_{M-3}(v^2),
\end{equation}
where similarly $P_M=1$ for $M\le 0$. It follows from the recursion relation that this polynomial is of order $S$ in
$v^2$.

Let $\tilde v_k^2$ be the roots of this polynomial; it can be shown that
these are all real. Let us furthermore fix the $S$ numbers $\tilde v_k$ to be all positive. 

The fermionic creation and annihilation operators $\Psi_k$ that solve the model were introduced in
\cite{fendley-fermions-in-disguise} as follows. Let $\chi_{M+1}$ be an operator which commutes with every $h_j$ except for
$h_M$, for which $\{h_M,\chi_{M+1}\}=0$. Then
\begin{equation}
  \label{Psidef}
  \Psi_{\pm k}=\frac{1}{N_k}T(\mp \tilde v_k) \chi_{M+1}  T(\pm \tilde v_k),
\end{equation}
where $N_k$ is a normalization factor whose precise value we will not use here, and we omit the subscript $M$ from the transfer matrix hereafter: $T(v) \equiv T_M(v)$.
It was shown in \cite{fendley-fermions-in-disguise}  that these operators satisfy the standard fermionic creation-annihilation algebra
\begin{equation}
  \{\Psi_{+k},\Psi_{+k'}\}=0,\quad  \{\Psi_{-k},\Psi_{-k'}\}=0,\quad
   \{\Psi_{+k},\Psi_{-k'}\}=\delta_{k,k'}
 \end{equation}
 and that the transfer matrix can be expressed using them as
 \begin{equation}
   \label{Tffd}
   T(v)=\prod_{k=1}^S (1-v\eps_k [\Psi_k,\Psi_{-k}]),
 \end{equation}
where we introduced {the quasi-particle energies $\eps_k$ as}
\begin{equation}
  \eps_k=\frac{1}{\tilde v_k}.
\end{equation}
This implies that the eigenvalues of the transfer matrix are
\begin{equation}
    \label{teig}
\Lambda(v)=  \prod_{j=1}^S(1\mp v\eps_k).
\end{equation}
Collecting the terms linear in $v$ we find that the eigenvalues of the Hamiltonian are
\begin{equation}
  E=\sum_{j=1}^S \mp \eps_k.
\end{equation}

The work \paulcite chose the $\chi_{M+1}$ operators to lie outside the FFD algebra; this is possible for the
representations treated there. However, we also introduced the minimal representations, for which every operator acting
in the given 
Hilbert space is generated by the FFD algebra.
Therefore, it is an important question, whether the free fermionic operators $\Psi_k$ can be constructed within these minimal
representations.
Alternatively, the question is whether the desired operator $\chi_{M+1}$ exists within the FFD algebra for the special
cases $M=6k$ 
and $M=6k+2$.

We show constructively that this is true. For $M=6k$ we can choose
\begin{equation}
     \chi_{M+1}=(ih_2h_3h_4)(ih_8h_9h_{10})\cdots (ih_{6k-4}h_{6k-3}h_{6k-2}),
 \end{equation}
and for $M=6k+2$ we can choose
\begin{equation}
     \chi_{M+1}=h_1h_4h_7\cdots h_{6k+1}.
 \end{equation}
 A direct check shows that these operators commute with every $h_j$ except with the last one $h_M$. These $\chi_{M+1}$
 are identical to the central elements that we found earlier for other values of $M$, but now we use them in the FFD
 algebras with $M=6k$ and $M=6k+2$.

 In \paulcite the $\chi_{M+1}$ are called boundary operators and they indeed act only on some spins
 at one of the boundaries of the chains. In contrast, in the minimal representations the $\chi_{M+1}$ 
 are not localized to the boundary, instead they spread out over the whole spin chain. The construction
 needs only their 
algebraic properties, which are satisfied by the choices above.

The fermionic operators $\Psi_{\pm k}$ are non-local in terms of the physical spin operators on the spin chain. This
means that the formula \eqref{Tffd} provides a factorization of the transfer matrix into very non-local
quantities. Thus, \eqref{Tffd} can not be used to build a quantum circuit with local unitary gates.

\subsection{Circuits in the FFD algebra}

\label{sec:ffdcirc}

In the following Sections we construct free fermionic circuits from the FFD algebra, using only local unitary gates.
In the analytic computations below we do
not specify the representations. However, for the numerical computations in Section \ref{sec:numerics} we make a choice
and use the minimal representations. These representations are somewhat irregular, in the sense that the representants
of  the $h_k$ generators are not homogenous in $k$, and for different values of $k$ they act on 1, 2, 3 or 4 sites. For
physical applications one could use the representation given by \eqref{basicrep2}, where each gate
acts on at most 3 sites. The advantage of this representation is that it allows for the addition of one more copy of the
FFD algebra, making the combined representation minimal. That way one can obtain homogeneous free fermionic circuits
based on the commuting set of generators $h_k$ and $\tilde h_k$.

Now we discuss two basic properties of the circuits \eqref{productdef}.

\begin{prop}
  The eigenvalues of the circuits built from the FFD algebra are reflection symmetric around the identity: if
  $\lambda=e^{i\phi}$, $\phi\in\valos$ is an 
  eigenvalue then $1/\lambda=e^{-i\phi}$ is also an eigenvalue.
\end{prop}
This can be proven most easily by using a modification of the representation \eqref{basicrep}. We can choose
$h_j=X_jX_{j+1}Y_{j+2}$, all the commutation relations are still satisfied, and we obtain a faithful representation.
In this representation the $h_j$ are
matrices with purely imaginary coefficients, which implies that the $u_j$ are real matrices, therefore they are members
of $SO(2^{L+2})$. The eigenvalues of an orthogonal matrix are either equal to unity or they come in pairs $e^{\pm i
  \eps}$, and this proves the proposition.

\bigskip

Later we will also need this simple statement:

\begin{prop}
  If a circuit built from the FFD algebra  has only 4 different eigenvalues (with uniform degeneracies), then it is free fermionic.
\end{prop}
It follows from the argument above that the 4 eigenvalues come in pairs
\begin{equation}
  e^{\pm i\phi_{a}},\quad e^{\pm i\phi_{b}}
  .
\end{equation}
These four eigenvalues are always of the form \eqref{fermioniclambda} with two eigenmodes, with pseudo-energies
\begin{equation}
  \eps_1=\frac{\phi_a+\phi_b}{2},\quad  \eps_2=\frac{\phi_a-\phi_b}{2}.
\end{equation}

\section{FFD algebra: a circuit from the transfer matrix}

\label{sec:ffdc1}

Our first example for a unitary quantum circuit in the FFD model comes from a special factorization of the transfer matrix
derived in \paulcite. First we repeat the relevant results of \paulcite and afterwards we specialize them to the unitary
case. The circuits we obtain here have a distinguished property: the unitary operator $\VV$ will commute
with a Hamiltonian of the form \eqref{Hdef}. This does not happen in other cases with other geometries; typically
the lowest order commuting charge is of higher order in the FFD algebra.

Let us then start with the Hamiltonian  as given by \eqref{Hdef} and let us construct the transfer matrix $T(v)$ as
described in the previous Subsection. It was shown in \paulcite that this transfer matrix can be factorized as follows
\begin{equation}
  \label{tfact}
  T_M(v)=G\cdot G^T,
\end{equation}
where
\begin{equation}
  \label{GPaul}
  G=g_1g_2\cdots g_{M-1}g_M
\end{equation}
The $g_j$ are local operators defined as
\begin{equation}
  \label{gdef}
  g_j=\cos(\phi_j/2)+h_j \sin(\phi_j/2).
\end{equation}
with the angles determined by the recursion
\begin{equation}
  \sin(\phi_{j+1})=-\frac{v\alpha_{j+1}}{\cos(\phi_{j})\cos(\phi_{j-1})},
\end{equation}
together with the initial condition $\phi_{0}=\phi_{-1}=0$. Furthermore, the notation $[.]^T$ stands for the operation defined in
\eqref{transposedef}, and it is not to be confused with the standard transposition\footnote{The paper \paulcite used the same notation $G^T$, and there it was meant to denote the usual
  transpose. The main parts of \paulcite dealt with the representation \eqref{basicrep}, where each $h_j$ is
  represented by a symmetric matrix. Therefore, this did not make any difference. On the other hand, it is clear from
  the derivations in  \paulcite that in other representations the operation \eqref{transposedef} should be used, and not the usual transpose.}.

The reader might wonder how the $2M$-fold product in  \eqref{tfact} can reproduce the formula \eqref{Tdef}. After all,
the definitions appear rather different.
The key is that there is a mechanism
which leads to various cancellations in \eqref{tfact}. Expanding each $g_j$ we obtain various products over the $h_j$
generators, and each product over $n$ generators appears 
exactly $2^n$ times, due to the possibilities of choosing the $h_j$ from the lower or the upper ``half'' of the product
\eqref{tfact}. The pre-factors coincide, but the ordering of the $h_j$ depends on the choices from where they are
taken. If there are two anti-commuting
operators in a certain product, then the whole sum cancels, because for every operator ordering there is exactly one
ordering which will eventually cancel it. This means that expanding the product \eqref{tfact} and summing everything up
leaves us with terms where the $h_j$ commute within every operator product. And this is indeed the defining property of
the charges for the transfer matrix.

Having understood the mechanism of why the two sides of \eqref{tfact} can coincide, the equality of the various terms
after the expansion still needs to be proven. To this order a recursive proof was provided in
\cite{fendley-fermions-in-disguise}, which we do not repeat here.

The resemblance of formulas \eqref{gdef} and \eqref{udef} is striking: the factorization \eqref{tfact} is structurally
very close to a quantum circuit. The only crucial difference is that the operators $g_j$ are Hermitian for real $v$, whereas we
are interested in unitary operators. 

It is then natural to set $v$ to be purely imaginary, so that the $g_j$ become proportional to 
unitary operators. The overall normalization of the Hamiltonian is arbitrary, and the couplings $\alpha_j$ appear only
in combination with the spectral parameter. Therefore we are free to set $v=i$, and then we substitute
\begin{equation}
  \phi_j=i\theta_j,\qquad \theta_j\in\valos.
\end{equation}
Then the recursion above becomes
\begin{equation}
  \label{rec2}
  \sinh(\theta_{j+1})=-\frac{\alpha_{j+1}}{\cosh(\theta_{j})\cosh(\theta_{j-1})},
\end{equation}
together with
\begin{equation}
  \label{gdef2}
  g_j=\cosh(\theta_j/2)+i \sinh(\theta_j/2) h_j.
\end{equation}
We see that in this case $T$ is proportional to a unitary operator:
\begin{equation}
  \label{vvdef1}
  T_M(i)=\mathcal{C} \VV,\qquad
  \VV=G\cdot G^T,
\end{equation}
with
\begin{equation}
  G=u_1u_2\cdots u_{M-1}u_M,
\end{equation}
where now
\begin{equation}
  u_j=\frac{g_j}{\sqrt{\cosh(\theta_j)}}=e^{i\varphi_j h_j}
\end{equation}
with
\begin{equation}
  \label{szgek}
  \tan(\varphi_j)=\tanh(\theta_j/2).
\end{equation}
The number $\CC$ is an irrelevant normalization factor, which can be collected from the differences in the overall
normalization of the $g_j$ and the $u_j$.

For completeness we also express the recursion directly in terms of the $\varphi_j$. Condition \eqref{szgek} implies
\begin{equation}
  \cosh(\theta_j)=\frac{1}{\cos(2\varphi_j)},\qquad \sinh(\theta_j)=\tan(2\varphi_j),
\end{equation}
and this leads to the recursion
\begin{equation}
  \label{rec2b}
\tan(2\varphi_{j+1})
  =-\alpha_{j+1}   \cos(2\varphi_{j})\cos(2\varphi_{j-1}).
\end{equation}
With this we have established that the unitary operator $\VV$ defined in \eqref{vvdef1} is free fermionic, and it commutes
with a Hamiltonian linear in the $h_j$. The eigenvalues of $\VV$ can be computed from the eigenvalues of the commuting
Hamiltonian, using the technology of \paulcite  described above.

The approach of \paulcite was to start from a Hamiltonian of the form \eqref{Hdef} and then to construct the transfer
matrix. Our point of view here is different: We are interested in the quantum circuits, therefore we are free to start
with the circuit $\VV$ with arbitrary angles $\varphi_j$. Then the computation proves that the circuit is free
fermionic, and it commutes with a local Hamiltonian, where the coefficients of that Hamiltonian are computed simply
from \eqref{szgek}-\eqref{rec2}. For completeness we present explicit formulas for this approach.

The polynomials $P_M(u)$  are obtained via the recursion \eqref{PMrec}, and this recursion is written in terms of the
angles $\varphi_j$ as 
\begin{equation}
  \label{PMrec2}
  P_M(v^2)=P_{M-1}(v^2)-v^2\left(\frac{\tan(2\varphi_{M})}{\cos(2\varphi_{M-1})\cos(2\varphi_{M-2})}\right)^2
  P_{M-3}(v^2).
\end{equation}
Specializing \eqref{teig} we find that the eigenvalues of the transfer matrix are
\begin{equation}
  \prod_{j=1}^S(1\mp i\eps_k)=\CC\times  \prod_{j=1}^S e^{\pm i \tilde \eps_k},\qquad
  e^{\pm i \tilde \eps_k}=\frac{1\mp i\eps_k}{\sqrt{1+\eps_k^2}},
\end{equation}
where $\tan(\tilde\eps_k)=\eps_k=1/{\tilde v_k}$ and $(\tilde v_k)^2$ are the roots of the polynomials $P_M(v^2)$ with $M$ being
the size of the FFD algebra. 

If we choose a homogeneous circuit with $\varphi_j=\varphi$, then the recursion becomes homogeneous except for the first
few steps. Explicitly we get
\begin{equation}
  \begin{split}
    P_0(v^2)&=1\\
    P_1(v^2)&=1-v^2 \tan^2(2\varphi)\\
   P_2(v^2)&=1-v^2 \sin^2(2\varphi)\frac{\cos^2(2\varphi)+1}{\cos^4(2\varphi)}
  \end{split}
\end{equation}
and
\begin{equation}
  P_M(v^2)=P_{M-1}(v^2)-v^2\frac{\sin^2(2\varphi)}{\cos^6(2\varphi)}P_{M-3}(v^2),\quad M\ge 3
  .
\end{equation}

We also compute the Hamiltonian which commutes with this homogeneous circuit. We find that 
almost all
coefficients of the Hamiltonian will be equal. After multiplication by some irrelevant global normalization factors, we
obtain the following Hamiltonian, which commutes with $\VV$:
\begin{equation}
 \label{rec3}
 H=\cos^2(2\varphi)h_1+\cos(2\varphi)h_2+\sum_{j=3}^M h_j
 .
\end{equation}
Here we used the relation \eqref{szgek} specialized to the homogeneous case. Note that this Hamiltonian is homogeneous
except the two terms at  the boundary, but in the limit $\varphi\to 0$ it becomes completely homogeneous, as expected.

It is important that the circuit $\VV$ includes every $u_j$ operator twice, but always with the same angle
$\varphi_j$. Choosing different angles for the two appearances would in general destroy the free fermionic nature. This
is most easily seen in the case of small $M$, for example at $M=4$. If the circuit would be free fermionic for every
choice of the $2M$ angles,  then we could
choose the first $M$ angles to be zero and it would imply that the circuit $u_1u_2u_3u_4$ is also free fermionic. Below
in Subsection \ref{sec:M4} we explicitly show that this is not the case.

\section{FFD algebra: circuits for small system sizes}

\label{sec:small}

Here we consider the FFD algebra with small values of $M$. We find that for the simplest cases $M=2$ and $M=3$ every
circuit is free fermionic, but for $M=4$ this is not true anymore. For $M=4$ we identify the circuits that are free fermionic.

\subsection{$M=2$}

In this case $h_1$ and $h_2$ are two anti-commuting operators, both of them squaring to the identity. They can be
faithfully represented by two single qubit Pauli operators, say $X$ and $Z$. Therefore, every product of the
unitaries $u_1$ and $u_2$ can be seen as just a rotation of a single qubit, and the eigenvalues are always of the
form $\lambda=e^{\pm i\phi}$. They can be computed simply from the algebra of one-site Pauli matrices.
Therefore, we have a single fermionic eigenmode.

\subsection{$M=3$}

Now we have three operators $h_{1,2,3}$, all of them squaring to the identity, and all of them anti-commuting with each
other. It is tempting to represent them with single spin Pauli  matrices, and then perform the computations within the
Pauli algebra. However, this is slightly misleading. 

For $M=3$ the FFD 
algebra has a central element $C=ih_1h_2h_3$, which is Hermitian and it squares to one. Its eigenvalues are $\pm
1$. Therefore, the FFD algebra can be represented on a single qubit only after we project to either one of the
subsectors with eigenvalue $c=\pm 1$. For $c=1$ a representation is given by
\begin{equation}
  h_1=X,\quad h_2=Y,\quad h_3=Z,
\end{equation}
while for $c=-1$ a representation is found as
\begin{equation}
  h_1=Z,\quad h_2=Y,\quad h_3=X.
\end{equation}
Within each sector we see that every product of the $u_j$ operators ($j=1,2,3$) will describe a rotation of a single
qubit, leaving us with a pair of eigenvalues  $e^{\pm i\theta}$. However, these eigenvalues can differ in the two
sectors, leaving us with a total of 4 eigenvalues 
\begin{equation}
  e^{\pm i\theta(c=1)},\quad \text{and} \quad e^{\pm i\theta(c=-1)}.
\end{equation}
It was explained in Subsection \ref{sec:ffdcirc} that a set of four inversion symmetric eigenvalues is always free fermionic,
therefore the circuits built on the FFD algebra with $M=3$ are always free fermionic, with at most 2
fermionic eigenmodes. 

The reader might wonder whether the two different choices above really lead to different eigenvalues. We demonstrate
this on a simple example. Let us consider the circuit
\begin{equation}
\VV=u_3u_2u_1.
\end{equation}
Knowing that this operator describes a single qubit rotation within each sector $c=\pm 1$, we can find the eigenvalues
by expanding the product and collecting the ``scalars'' within the given sector. This will give the eigenvalues $e^{\pm i\phi}$, where
\begin{equation}
  \cos(\varphi_1)\cos(\varphi_2)\cos(\varphi_3)+  c\sin(\varphi_1)\sin(\varphi_2)\sin(\varphi_3)=\cos(\phi).
\end{equation}
For generic values of $\varphi_{1,2,3}$ this will result in different values of $\phi$ depending on $c=\pm 1$. 

\subsection{$M=4$}

\label{sec:M4}

In this case we find the following three Hermitian central elements:
\begin{equation}
  C_1=ih_1h_2h_3,\qquad C_2=ih_2h_3h_4,\qquad C_3=h_1h_4.
\end{equation}
They are not independent from each other, because $C_3=C_1C_2$. Furthermore it can be seen that there are no other
central basis elements in the algebra.

Every one of the three central elements square to one, therefore their
eigenvalues are $\pm 1$.  For the eigenvalues we will use the notation $c_{1,2,3}$.
This implies that the action of every operator can be split into 4 different
sectors, corresponding to the eigenvalues of two chosen central elements, say $C_1$ and $C_3$.

Once the sector is specified, the operators can be identified once again with single spin Pauli matrices, in the
following way. If we represent $h_1$ with $Z$, then we also get $h_4=c_3 Z$.
Furthermore, from the eigenvalue $c_1=\pm 1$ we can see that the
remaining two generators $h_2$ and $h_3$ can be represented with $Y$ and $X$, or $X$ and $Y$, respectively.

With this we have established that the spectral problem of every operator splits into 4 single qubit problems. For our
main problem we are considering unitary operators, so we are faced with 4 different single qubit rotations, each of
which have a pair of eigenvalues $e^{\pm i\phi}$. This means that in the general case a circuit in the FFD algebra with
$M=4$ will have 8 different eigenvalues. 

In the generic case these 8 eigenvalues will not be free fermionic anymore. This can be seen already at the simplest
examples with 4 different gates. Consider for example the circuit
\begin{equation}
  \VV=u_4u_3u_2u_1.
\end{equation}
In this case a direct computation of the single qubit rotations together with the proper choices dictated by the
eigenvalues $c_1$ and $c_3$ give the eigenvalues $e^{\pm i\phi}$ with the angle $\phi$  determined by
\begin{equation}
  \begin{split}
    \cos(\phi)=&\cos(\varphi_1)\cos(\varphi_2)\cos(\varphi_3)\cos(\varphi_4)
    +c_1\sin(\varphi_1)\sin(\varphi_2)\sin(\varphi_3)\cos(\varphi_4)+\\
&    +c_2\cos(\varphi_1)\sin(\varphi_2)\sin(\varphi_3)\sin(\varphi_4)
     -c_1c_2\sin(\varphi_1)\cos(\varphi_2)\cos(\varphi_3)\sin(\varphi_4)
  \end{split}
\end{equation}
with $c_{1,2}=\pm 1$. 
Altogether we find 4 different values for $\phi$, without any obvious relations between them, unless the $\varphi_j$
angles are also fine 
tuned. This means that this circuit will have 8 different eigenvalues, and its spectrum is not free fermionic.

Let us now consider again the circuit of the previous Subsection for $M=4$, in light of the present discussion. This
circuit is given by
\begin{equation}
  \label{vm4}
  \VV=G\cdot G^T, \qquad G=u_1u_2u_3u_4.
\end{equation}
A direct expansion of \eqref{udef} gives $2^8$ terms in total, but due to the ordering of the gates most terms just
cancel. In the end we are left with contributions from the operators $h_{j}$, $j=1,\dots,4$ and $h_1h_4$. All other
operator products cancel, as expected from the results discussed above. In this simple case the central
element $h_1h_4$ is actually proportional to the charge $Q^{(2)}$ defined in \eqref{Q2def}.

Now we can understand the free fermionic nature of the circuit \eqref{vm4} directly, without employing the full
machinery of \paulcite. If the central elements $C_1$ and $C_2$ do not appear in the expansion of the circuit, then its
action can be split into two sectors according to the eigenvalues of $C_3$ only. This will result in 4 different
eigenvalues for this unitary circuit, which means eventually that the circuit is free fermionic.

Although this example is rather trivial, it gives some ideas for finding other free fermionic circuits. In this special
case of $M=4$ we should look for circuits for which only one central element appears in the expansion of the operator
$\VV$.  The central element $C_3$ is a product of commuting operators, therefore there is no
direct way to cancel it in the expansion of the circuits. This leaves us with the goal of cancelling 
$C_1$ and $C_2$,  which can actually be achieved by other geometries as well. An
example is given by
\begin{equation}
  \label{vm4b}
  \VV=G\cdot G^T,\qquad G=u_1u_4u_3u_2.
\end{equation}
A direct expansion of the gates shows that the structure of \eqref{vm4b} is similar to that of \eqref{vm4}: the same
operators appear, albeit with different coefficients. If the angles $\varphi_j$, $j=1,\dots,4$ are chosen to be the
same, then the two circuits are not isospectral. Nevertheless both include only $C_3$, and not $C_1$ or $C_2$ in their
expansion, therefore they are both free fermionic.

We can actually prove the following:
\begin{thm}
  For $M=4$ every circuit of the form $\VV=G\cdot G^T$ is free fermionic, irrespective of the length of the operator
  product in $G$. 
\end{thm}
\begin{proof}
  To prove this theorem, notice that $\VV^T=\VV$, therefore its expansion can only have operator products which also
satisfy this symmetry. For $M=4$ the only possible choices are (apart from the identity):
\begin{equation}
  h_1,\ h_2,\ h_3,\ h_4,\ h_1h_4,\ h_1h_2h_4,\ h_1h_3h_4.
\end{equation}
A direct check shows that the other operator products are anti-symmetric. The operator $C_3=h_1h_4$ is central, so the Hilbert
space splits into two sectors according to the eigenvalue $c_3=\pm 1$. Substituting $h_4=c_3h_1$ we see that $\VV$ can
be expanded in both sectors as a sum of only the following four terms:
\begin{equation}
  a+i(b_1h_1+b_2h_2+b_3h_3),
\end{equation}
where the real coefficients $a$ and $b_{1,2,3}$ depend on the parameters of the circuit and also the value of $c_3$. The
operator above has only two different eigenvalues $e^{\pm i\phi}$, which are given simply by
\begin{equation}
  \cos(\phi)=a,
\end{equation}
therefore they are independent from the eigenvalue of the remaining central
element $C_1$. This proves that the $\VV$ has only four different eigenvalues, each of them with double degeneracy,
therefore the circuit is free fermionic.
\end{proof}

It is important that the structure $\VV=G\cdot G^T$ is sufficient, but not necessary to have a free fermionic circuit. For example
we could transform the circuit using the cyclic permutations, leaving the spectrum intact. This would imply that other
operator products also appear in the expansion, and it would become more difficult to recognize the structure of the
circuit.

\section{FFD algebra: Numerical results}

\label{sec:numerics}

In this Section we describe numerical results, which show that certain regular geometries are free fermionic for increasing
values of $M$. First we present a simple numerical algorithm which determines whether a certain circuit is free
fermionic. Afterwards we describe the results that we obtained. A selection of our actual numerical data is presented later
in Appendix \ref{sec:appnum}.

\subsection{Numerical methods}

The main idea is to compute the eigenvalues $\lambda_j$ of a given $\VV$ numerically, and then to compute all possible
ratios $\lambda_j/\lambda_k$ (without any restrictions on $j$ and $k$).

If there are $2^n$ distinct values of
$\lambda_j$, then a free fermionic spectrum would 
give a total number of $3^n$ distinct values for the ratios $\lambda_j/\lambda_k$. This can be seen easily by
considering the form \eqref{fermioniclambda} of the eigenvalues: computing a ratio we find
\begin{equation}
  \frac{\lambda_j}{\lambda_k}=\prod_{\ell=1}^n e^{i\kappa_\ell \eps_\ell},\quad\text{with}\quad \kappa_\ell=-2, 0, 2.
\end{equation}

Meanwhile, a generic non-free spectrum with conjugation
symmetry would give $2\cdot
4^{n-1}+1$ distinct values. For $n=1$ and $n=2$ every
conjugation symmetric spectrum is free fermionic, correspondingly we find the coincidences
\begin{equation}
  3^n=2\cdot 4^{n-1}+1,\qquad n=1, 2.
\end{equation}
However, the predicted numbers on the two sides are different for any $n>2$, which makes a clear distinction possible.

In our numerical computations we choose random angles $\varphi_k$ for the unitary gates, and
we compute the total number of distinct eigenvalues, and the number of distinct ratios.

\subsection{Circuits that are not free fermionic}

Before detailing our numerical results for the free fermionic circuits, we give here a list of simple geometries which we found
to be not free.

Our first example is simply the product
\begin{equation}
  \VV=u_M\dots u_2u_1.
\end{equation}
In \ref{sec:M4} it was proven that this circuit is not free fermionic already for $M=4$, and this persists also for
larger $M$.

We can also consider the circuit (assuming $M=2k$)
\begin{equation}
  \VV=(u_M\dots u_4u_2)(u_{M-1}\dots u_3u_1).
\end{equation}
In the case of nearest neighbour interacting spin chains, this geometry corresponds to the standard ``brickwork circuit''.
 However, in the present case we find that it is not free fermionic. This is clear already at
$M=4$, as explained in \ref{sec:M4}.

An alternative idea is to construct a  ``brickwork circuit'' adapted to the representation \eqref{basicrep}, where every
gate $u_j$ acts on three neighbouring sites. Here the idea is to build a circuit with spatial period three. 
For example, in the case of $M=3k$ we could take  (see Fig. \ref{fig:cbw3})
\begin{equation}
\VV= (u_M\dots u_6u_3)(u_{M-1} \dots u_5u_2) (u_{M-2}\dots u_4u_1).
\end{equation}
An advantage of this circuit is, that within each parenthesis every gate commutes. 
However, this circuit isn't free fermionic either. This can be seen at $M=6$ by choosing $\varphi_5=\varphi_6=0$, in which
case the circuit simplifies to one with $M=4$, for which we have already proven that it it not free fermionic.

\subsection{Free fermionic circuits}

Despite the counter-examples above, it is possible to find free fermionic circuits for larger $M$, simply by
trial. However, we are interested in those cases which can be constructed for increasing values of $M$. 
Our main idea is to assume the structure
\begin{equation}
  \label{VVform}
\VV=G\cdot G^T,
\end{equation}
where $G$ is a some finite operator product, different from the one in eq. \eqref{GPaul}, which was derived in \paulcite
for the factorization of the transfer matrix.

The key reasons why such a structure could work is that it produces a lot of cancellations: only those operator
products appear in the expansion which are symmetric with respect to the operation $[.]^T$. In the special case of the
factorized transfer matrix of \paulcite (detailed in Section \ref{sec:ffdtm}) this is restricted to products of
commuting operators. In other cases non-commuting operators can also appear, but the example of $M=4$ shows that the
restriction of symmetry is sometimes enough to ensure the free fermionic nature.

Simple trials show that for relatively small values of $M$ the form above is likely to give a free fermionic circuit. In
fact, for small values of $M$ it is difficult to find a $G$, for which $\VV$ is not free fermionic.

The smallest $M$ for which we found a non-free circuit of this form is $M=8$. The example is
\begin{equation}
  G=u_1  u_4 u_3  u_5   u_6  u_8,
\end{equation}
which leads to a $\VV$ that is {\it not} free fermionic; an example for the numerical eigenvalues is presented in
Appendix \ref{sec:appnum}.

For the intermediate values $4<M<8$ we could not find
a circuit of the above form, 
which is not free fermionic. This lead us to the following:

\begin{conj}
  For $4<M<8$ every circuit of the form $\VV=G\cdot G^T$ is free fermionic, with $G$ being an arbitrary finite product of
  the local gates.
\end{conj}

We tested this conjecture by randomly selecting operator products for $G$, with varying length, and with random angles
$\varphi_j$. We tried a couple of thousands of $G$ for each length of the operator product, and in every case we found a
free fermionic spectrum. However, our search was not exhaustive, therefore the above statement is only a conjecture.

Finally, let us turn to the constructions that work for increasing values of $M$. We found that
multiple types of  regular geometries for $G$ do lead to free fermionic circuits. 

{One possibility is to take a modified staircase construction. The circuit for $M=2k$ this is given by}
\begin{equation}
  \label{G1}
  G=
  (u_1u_3\dots u_{M-1})(u_2u_4\dots u_{M}).
\end{equation}
{This circuit has linear depth, because the neighbouring gates $u_j$ and $u_{j+2}$ do not commute with each other.}

{It is desirable to construct circuits which have linear depth. 
A natural idea is to take the brickwork structure (assuming $M=3k$)}
\begin{equation}
  \label{G2}
G=(u_1u_4\dots u_{M-2}) (u_2u_5\dots u_{M-1}) (u_3u_6\dots u_M).
\end{equation}
For this choice the gates within each parenthesis commute with each other. If we consider the representation
\eqref{basicrep} then the gates in each paranthesis do not overlap, therefore $G$ has depth 3 and the circuit $\VV$ has
depth 6.
However, for other representations, the
gates within each parenthesis might still overlap with each other, even though they commute, and then the depth becomes
linear again.

We can extend these ideas, and we could also take for $M=4k$
\begin{equation}
  \label{G3}
  G=(u_1u_5\dots u_{M-3})(u_2u_6\dots u_{M-2})(u_3u_7\dots u_{M-1})
  (u_4u_8\dots u_{M}).
\end{equation}
{In the representation \eqref{basicrep} we obtain a circuit of depth 8, but for the minimal representation this circuit
is still of linear depth.}

Numerical tests show that all of these choices for $G$ lead to fermionic circuits, we thus have the following:

\begin{conj}
  A circuit of the form $\VV=G\cdot G^T$ with $G$ given by \eqref{G1}, \eqref{G2} or \eqref{G3} is free fermionic.
\end{conj}

We numerically tested this conjecture for $M\le 12$ and found that in this range it always holds. We believe that for
this many gates there can not be any coincidences, therefore the conjecture seems well justified. However, we can not
prove it at the moment. 

Concrete numerical data is presented in Appendix \ref{sec:appnum}. 

\subsection{Commuting charges}

\label{sec:commc}

In Subsection \ref{sec:posing} above we asked Question 2, which is about the existence of low order Hermitian charges
commuting with the circuits. For the new free fermionic circuits this question is open. In the following we present some
preliminary results, obtained with the computer program \textit{Mathematica}. 

Let us first describe the difficulty behind finding the charges. The first problem is that  a brickwork circuit with
multiple ``layers'' will likely 
increase the order the lowest charge in the $h_k$ generators. Already for simplest case of the Ising algebra we
established in \ref{sec:ising} that a brickwork circuit with two layers results in a charge that is a second order
polynomial in the $h_k$. This order is likely to increase for more complicated circuits. The second difficulty comes
from the boundary terms, that appear in the expected charges, making the extraction of the ``bulk'' part of the
commuting charges more difficult.

Preliminary computations using \textit{Mathematica} show that there are indeed commuting charges for the circuits that
we introduced above, but at present we have closed form results (applicable for increasing values of $M$) only for
limited cases.

For example consider the circuit $\VV=G\cdot G^T$ with $G$ given by \eqref{G1} with $M=2k$. For simplicity let us choose
all angles to be equal, thus $\varphi_j=\varphi$.  Then we find the
lowest order commuting charge
\begin{equation}
  \begin{split}
    Q=& \cos^3(2\varphi) h_1+ \cos(2\varphi) h_2+\\
&  (\cos^2(2\varphi) h_3+  h_4)+  (\cos^2(2\varphi) h_5+  h_6)+\dots+ (\cos^2(2\varphi) h_{M-3}+  h_{M-2})+\\
   & \cos^2(2\varphi) h_{M-1}   + \cos(2\varphi) h_{M}+\\
   &    \sin^2(2\varphi)\cos(2\varphi)  h_2h_3h_5 +  \\
   &\sin^2(2\varphi)( h_4h_5h_7+ h_6h_7h_9+\dots+
     h_{M-4}h_{M-3}h_{M-1}).
  \end{split}
\end{equation}
We can see that there is a bulk part which is homogeneous, but multiple boundary terms also appear. 
Note that in the limit $\varphi\to 0$ the third order terms disappear, and we obtain simply $\sum_j h_j$.

Numerical investigation shows that this Hermitian charge itself is free fermionic.
However, it is fine tuned to commute with the circuit $\VV$, and its terms can not be completely independent from each
other. Therefore, it does not immediately fall into the family of models treated by the graph theoretical framework of
\cite{fermions-behind-the-disguise,unified-graph-th}. 

Direct computations for small sizes show that the charge $Q$ appears in the direct expansion of $\VV$. This points to
the possibility of extracting higher charges from the expansions, and perhaps to construct a spectral parameter
dependent transfer matrix for this case too. Alternatively, a transfer matrix might be constructed from $Q$ directly. At
present it is not clear which method will eventually lead to the full proof of the free fermionic nature for this circuit.

\section{A free fermionic circuit in the generalized model}

\label{sec:FP}

In this Section we treat the model introduced in \cite{sajat-FP-model}, {which is a generalization of the FFD model,}
defined on a spin chain of length $L+2$ with the Hamiltonian
\begin{align}
  \label{FPH}
    H=C_{1}+\sum_{j=2}^{L}(A_{j}+B_{j}+C_{j}),
\end{align}
where 
\begin{align}
  \label{eq:FP-basic-op}
A_{j}=a_{j-1}a_{j}t_{j}^{A}
,\quad 
B_{j}=b_{j-1}b_{j}t_{j}^{B}
,\quad C_{j}=a_{j}b_{j}t_{j}^{c}.
\end{align}
Here $a_j,b_j$ with $j=1,\dots,L$ are real coupling constants, and $t_j^{A,B,C}$ are operators given by
\begin{align}
t_{j}^{A}=X_{j-1}X_{j}Z_{j+1}
,\quad 
t_{j}^{B}=Z_{j-2}Y_{j-1}Y_{j}
,\quad 
t_{j}^{C}=Z_{j-1}Z_{j+1}
.
\end{align}
We used the conventions of \cite{sajat-FP-model}, such that the sites of the chain are indexed as $0, 1, \dots, L,
L+1$.

In~\cite{sajat-FP-model}, it was demonstrated that the Hamiltonian~\eqref{FPH} and its transfer matrix are
free-fermionic; we review the key statements below in Subsection \ref{sec:genkey}.
In subsection~\ref{subsec:factorization-FP} we show that the transfer matrix of the generalized model can be
  factorized in a similar manner to the FFD case~\eqref{tfact}. Subsection~\ref{subsec:quantum-circuit-FP}
  demonstrates that the factorized transfer matrix can be expressed as a unitary quantum circuit with local gates.

If we set $a_j\equiv 0$  or $b_j\equiv 0$ then only the terms $B_j$ or $A_j$ survive in the Hamiltonian,
respectively. This implies that in both of these two limits we recover the FFD Hamiltonian treated earlier. In the case
$a_j=b_j=1$ we obtain the model equivalent to the one treated in \cite{gyuri-susy-1}, which in turn is equivalent to
the even earlier model of \cite{cooper-anyon} (see \cite{gyuri-susy-2}). Therefore, the above Hamiltonian interpolates
between previously discovered models, that appeared unrelated at that time.

{This model is important for the purely theoretical questions of free fermion solvability.}
The relevance of the model lies in the fact that the Hamiltonian does not directly fit into the framework of
\cite{fermions-behind-the-disguise,unified-graph-th}. There are two differences. First, the frustration graph of the
model does not satisfy the conditions given in \cite{fermions-behind-the-disguise,unified-graph-th} for free fermion
solvability. Second, the pre-factors for the different terms in the Hamiltonian are not independent from each other. In
fact, for a given $L$ there are $3L-2$ terms in the Hamiltonian, but only $2L$ coupling constants. It was remarked in
\cite{fermions-behind-the-disguise,unified-graph-th} that their conditions are sufficient but not necessary for the free
fermionic solvability, and that fine tuned cases do exist which are free fermionic only if there are special relations
between the pre-factors in the Hamiltonian. The model of \cite{sajat-FP-model} presents one such case.

The coupling constants in the model are arranged such that there is a special relation between the operators:
\begin{equation}
  A_jB_j+C_{j-1}C_j=0,
  \label{abcrel}
\end{equation}
and we also have
\begin{equation}
  [A_j,B_j]=[C_{j-1},C_{j}]=0.
\end{equation}
The relation \eqref{abcrel} is one of the cornerstones of the solvability of the model: it was shown in  \cite{sajat-FP-model} that
certain extended frustration graphs can be constructed if one takes into account this condition, such that the new
graphs already satisfy the requirements of \cite{fermions-behind-the-disguise}.

{Further commutation relations can be found in Appendix \ref{app:FP-commutation}.}

\subsection{Transfer matrix and free fermions in the generalized model}

\label{sec:genkey}

Here we review the key result of~\cite{sajat-FP-model}.
  The generalized model has a family of commuting extensive charges, which can be summed into a transfer matrix, and
  this transfer matrix is defined via a recursion relation in the volume $L$~\cite{sajat-FP-model}. The transfer matrix
  plays a crucial role in constructing the free fermionic operators, we review the key statements below.
 
We will use three different transfer matrices and set up recursion relations for them. First, we will use the
notation $T_L(v)$ for the transfer matrix {of the generalized model} acting on the $L$-site Hilbert space i.e.\ $\alpha_j=\beta_j=0$ for $j>L$ and
$j<1$. We also define $T_L^A(v)$ such that we also set $\beta_L=0$, and $T_L^B(v)$ such that $\alpha_L=0$ instead.
Then the transfer matrices for varying lengths obey the recursion relations 
\begin{align}
T^{}_L(v) &= T^A_{L}(v) + T^B_{L}(v) -T^{}_{L-1}(v) -v C_L T^{}_{L-2}(v)
,
\label{recursion}
\\ 
T^A_{L+1}(v) &= T^{}_{L}(v) -uA_{L+1} T^B_{L-1}(v)
,
\label{recursion-A}
\\
T^B_{L+1}(v) &= T^{}_{L}(v) - v B_{L+1} T^A_{L-1}(v)
,
\label{recursion-B}
\end{align}
for $L> 0$, with the initial conditions given by $T^{}_0=T^{}_{-1}=T^A_0=T^A_1=T^B_0=T^B_1=1$. 

It was proven in \cite{sajat-FP-model} that the transfer matrices form a commuting family:
\begin{equation}
  [T_L(u),T_L(v)]=0, \label{FPcomm}
\end{equation}
and the linear term in the spectral parameter is once again the Hamiltonian~\eqref{FPH}.

It was also proven in  \cite{sajat-FP-model} that the transfer matrix satisfies the inversion relation
\begin{equation}
  T_L(v)T_L(-v)=P_L(v^2),
\end{equation}
where $P_L(v^2)$ is a polynomial of order $S'=[(L+1)/2]$ in $v^2$. It is given by the following recursion relations analogous to
\eqref{recursion}-\eqref{recursion-B}, concerning three families of polynomials $P_L(v^2)$, $P_L^A(v^2)$, and $P_L^B(v^2)$: 
\begin{equation}
\begin{split}
  P_L(v^2) &= P^A_{L}(v^2) + P^B_{L}(v^2) -P_{L-1}(v^2) - \big(ua_L b_L\big)^2 P_{L-2}(v^2),
  \\ P^A_{L+1}(v^2) &= P_{L}(v^2) - \big(ua^{}_La^{}_{L+1}\big)^2 P^B_{L-1}(v^2),\\
  P^B_{L+1}(v^2) &=P_{L}(v^2) - \big(ub_Lb_{L+1}\big)^2 P^A_{L-1}(v^2),
\label{Precursion}
\end{split}
\end{equation}
for $L> 0$, defining $P_0=P_{-1}=P^A_0=P^A_1=P^B_0=P^B_1=1$. 

Let us denote by  $(\tilde v_k)^2$ with $k=1,\dots,S'$ the roots of the polynomials $P_L(v^2)$. They will give the fermionic
one-particle energies via the usual relation $\eps_k=1/\tilde v_k$. The free fermionic raising and lowering operators
are defined as
\begin{equation}
  \Psi_{\pm k}=\frac{1}{N_k} T(\mp \tilde v_k)\alpha_{L+1} T(\pm \tilde v_k)
  .
\end{equation}
Here $\alpha_{L+1}$ is an edge operator, which commutes with every term $t^{A,B,C}_j$, except for $t^A_{L}$ and $t^C_L$,
with which it anti-commutes. One possible choice is $\alpha_{L+1}=X_{L+1}$.  The number $N_k$ is a normalization factor
whose precise value is not important for our purposes.

It was shown in \cite{sajat-FP-model} that these operators satisfy the standard fermionic algebra, and that they also
satisfy
\begin{equation}
  [H,\Psi_{\pm k}]=\pm \frac{2}{\tilde v_k}\Psi_{\pm k},
\end{equation}
therefore they are indeed raising and lowering operators. Furthermore, it can be shown that the transfer matrix
factorizes into the fermionic eigenmodes as
\begin{equation}
  \label{Tgenf}
   T_L(v)=\prod_{k=1}^{S'} (1-v\eps_k [\Psi_k,\Psi_{-k}]),
\end{equation}
where we defined again $\eps_k=1/\tilde v_k$.  This formula is analogous to \eqref{Tffd}; the two differences are that now
there are $S'$ eigenmodes, and that both the transfer matrix and the fermions are different operators than in the case
of the FFD model.

The factorization \eqref{Tgenf} is non-local in nature, because the fermionic eigenmodes are non-local.

  \subsection{Factorization of the transfer matrix with local operators}
\label{subsec:factorization-FP}

Now we will show that the transfer matrix obtained in~\cite{sajat-FP-model} can be factorized using
 local gates in a similar manner with the FFD case~\eqref{tfact}, which will directly prove the commutativity via
 \eqref{FPcomm}. This factorization is essential in order to have a unitary quantum circuit with local gates.
 We stress that this factorization of the 
transfer matrix is a new result of this work, and it does not follow trivially from  previously published results.

Our key result is the following: The transfer matrix factorizes as
\begin{align}
 T_{L}(v) =G_L\cdot G^{\top}_L
  \,,
  \label{eq:FPfact}
\end{align}
where, for $L>1$, we define
\begin{align}
  \label{GMek2}
 G_L = & g_{1}^{C}  (g_{2}^{B}g_{2}^{A}g_{2}^{C})(g_{3}^{B}g_{3}^{A}g_{3}^{C})\cdots(g_{L}^{B}g_{L}^{A}g_{L}^{C})
 ,
  \\
 G^{\top}_L = & (g_{L}^{C}g_{L}^{A}g_{L}^{B})(g_{L-1}^{C}g_{L-1}^{A}g_{L-1}^{B}) \cdots (g_{2}^{C}g_{2}^{A}g_{2}^{B}) g_{1}^{C}
 ,
  \end{align}
 and we also define $G_1 = G_1^\top = g_1^C$.
 Here $g_{j}^{A}$, $g_{j}^{B}$ and $g_{j}^{C}$ are linear operators localized around site $j$:
\begin{align}
 g_{j}^{A}=\cos\frac{\phi_{j}^{A}}{2}+t_{j}^{A}\sin\frac{\phi_{j}^{A}}{2},\quad g_{j}^{C}=\cos\frac{\phi_{j}^{C}}{2}+t_{j}^{C}\sin\frac{\phi_{j}^{C}}{2},\quad g_{j}^{B}=\cos\frac{\phi_{j}^{B}}{2}+t_{j}^{B}\sin\frac{\phi_{j}^{B}}{2}
 .
\end{align}
The angles are determined from the following recursion equation:
\begin{align}
    \label{eq:recursion-angle-FP-a}
 \sin\phi_{j+1}^{A}
 =&
 -\frac{v a_{j} a_{j+1}}{\cos\phi_{j-1}^{C}\cos\phi_{j-1}^{A}\cos\phi_{j}^{A}\cos\phi_{j}^{B}\cos\phi_{j}^{C}} 
 \qquad \text{for $j\ge 1$}
 ,\\
    \label{eq:recursion-angle-FP-b}
 \sin\phi_{j+1}^{B}
 =&
 -\frac{v b_{j} b_{j+1}}{\cos\phi_{j-1}^{C}\cos\phi_{j-1}^{B}\cos\phi_{j}^{A}\cos\phi_{j}^{B}\cos\phi_{j}^{C}}
 \qquad \text{for $j\ge 1$}
 ,\\
    \label{eq:recursion-angle-FP-c}
 \sin\phi_{j+1}^{C}
 =&
 -\frac{v a_{j+1} b_{j+1}}{\cos^{2}\phi_{j}^{C}\cos\phi_{j}^{A}\cos\phi_{j}^{B}\cos\phi_{j+1}^{A}\cos\phi_{j+1}^{B}}
 \qquad \text{for $j\ge 0$}
 ,
\end{align}
together with the initial condition
\begin{align}
    &
    \phi_{0}^{A} = \phi_{0}^{B} = \phi_{0}^{C} = \phi_{1}^{A} = \phi_{1}^{B} = 0
    .
    \label{eq:initial-cond}
\end{align}
 The detail of the proof of the factorization of the transfer matrix~\eqref{eq:FPfact} is given in Appendix \ref{app:proof-factorization-FP}.
{The transfer matrices $T^A_L(v)$ and $T^B_L(v)$ defined in the previous Subsection are factorized as ($L > 1$)}
\begin{align}
  \label{eq:FPfact-A}
 T^A_{L}(v) &= G_{L-1} \left(g_{L}^{A}\right)^2 G^{\top}_{L-1}
 ,
  \\
  \label{eq:FPfact-B}
 T^B_{L}(v) &= G_{L-1} \left(g_{L}^{B}\right)^2 G^{\top}_{L-1}
 .
\end{align}

The angles defined by~\eqref{eq:recursion-angle-FP-a}-\eqref{eq:recursion-angle-FP-c} satisfy the following relation:
\begin{align}
  \label{another-relation-FP}
 \sin\phi_{j-1}^{C}\sin\phi_{j}^{C}=&\tan\phi_{j}^{A}\tan\phi_{j}^{B}
 \qquad \text{for $j\ge 1$}
 .
\end{align}
The proof of this relation is given in the Appendix~\eqref{app:proof-another-relation}.

\subsection{The free fermiomic quantum circuit}
\label{subsec:quantum-circuit-FP}

Here we show that the factorized transfer matrix can be expressed as an unitary quantum circuit.
For real values of $v$ the transfer matrix $T(v)$ is Hermitian, this follows from the recursion relations \eqref{recursion}-\eqref{recursion-B}. It follows from the
recursion relations for the angles, that the operators $g_j^{A,B,C}$ are also Hermitian for real $v$, at least in a
neighbourhood of $v=0$. Correspondingly, if we set $v=iy$ with $y\in\valos$, then every $g_j^{A,B,C}$ becomes
proportional to a unitary matrix. In order to see this, let use set $\phi_j^{x}=i\theta_j^x$ with $x=A,B,C$. We get the
operators
\begin{align}
  \label{gunit}
    g_{j}^{A}=\cosh\frac{\theta_{j}^{A}}{2}+it_{j}^{A}\sinh\frac{\theta_{j}^{A}}{2},\quad
    g_{j}^{C}=\cosh\frac{\theta_{j}^{C}}{2}+it_{j}^{C}\sinh\frac{\theta_{j}^{C}}{2},\quad
    g_{j}^{B}=\cosh\frac{\theta_{j}^{B}}{2}+it_{j}^{B}\sin\frac{\theta_{j}^{B}}{2}
    ,
  \end{align}
with the recursion relations
\begin{align}
    \sinh\theta_{j+1}^{A}
    =&
    -\frac{y a_{j} a_{j+1}}{\cosh\theta_{j-1}^{C}\cosh\theta_{j-1}^{A}\cosh\theta_{j}^{A}\cosh\theta_{j}^{B}\cosh\theta_{j}^{C}} 
    \qquad \text{for $j\ge 1$}
    ,\\
    \sinh\theta_{j+1}^{B}
    =&
    -\frac{y b_{j} b_{j+1}}{\cosh\theta_{j-1}^{C}\cosh\theta_{j-1}^{B}\cosh\theta_{j}^{A}\cosh\theta_{j}^{B}\cosh\theta_{j}^{C}}
    \qquad \text{for $j\ge 1$}
    ,\\
    \sinh\theta_{j+1}^{C}
    =&
    -\frac{y a_{j+1} b_{j+1}}{\cosh^{2}\theta_{j}^{C}\cosh\theta_{j}^{A}\cosh\theta_{j}^{B}\cosh\theta_{j+1}^{A}\cosh\theta_{j+1}^{B}}
    \qquad \text{for $j\ge 0$}
    ,
\end{align}
together with the initial conditions
\begin{align}
    &
    \theta_{0}^{A} = \theta_{0}^{B} = \theta_{0}^{C} = \theta_{1}^{A} = \theta_{1}^{B} = 0
    .
    \label{eq:initial-cond2}
\end{align}
The operators $t^{A,B,C}_j$ are all Hermitian, therefore in this case the operators in \eqref{gunit} are proportional to
unitary matrices. They describe local gates acting on two or three sites on the chain.

We can now introduce a further parametrization, similar to the case of the FFD model, such that we directly get unitary
matrices of the form
\begin{equation}
  u_j^{A}=e^{i\varphi_j^A t^A_j},\quad u_j^{B}=e^{i\varphi_j^B t^B_j},\quad u_j^{C}=e^{i\varphi_j^C t^C_j}
\end{equation}
The identification is simply
\begin{equation}
  \label{uj}
\tanh\left(\frac{\theta^x_j}{2}\right)=\tan(\varphi^x_j),\qquad  u_j^x=\frac{g^x_j}{\sqrt{\cosh(\theta_j^x)}},\qquad x=A, B, C
\end{equation}
and the recursion for the new angles is given by
\begin{align}
    \tan 2\varphi_{j+1}^{A}
    =&
       -y a_{j} a_{j+1}
       \cos 2\varphi_{j-1}^{C}\cos 2\varphi_{j-1}^{A}\cos 2\varphi_{j}^{A}\cos 2\varphi_{j}^{B}\cos 2\varphi_{j}^{C} 
    \qquad \text{for $j\ge 1$}
    ,\\
    \tan 2\varphi_{j+1}^{B}
    =&
       -y b_{j} b_{j+1}
       \cos 2\varphi_{j-1}^{C}\cos 2\varphi_{j-1}^{B}\cos 2\varphi_{j}^{A}\cos 2\varphi_{j}^{B}\cos 2\varphi_{j}^{C}
    \qquad \text{for $j\ge 1$}
    ,\\
    \tan 2\varphi_{j+1}^{C}
    =&
       -y a_{j+1} b_{j+1}
       \cos^{2}2\varphi_{j}^{C}\cos 2\varphi_{j}^{A}\cos 2\varphi_{j}^{B}\cos 2\varphi_{j+1}^{A}\cos 2\varphi_{j+1}^{B}
    \qquad \text{for $j\ge 0$}
    ,
\end{align}
together with the initial conditions
\begin{align}
    &
    \varphi_{0}^{A} = \varphi_{0}^{B} = \varphi_{0}^{C} = \varphi_{1}^{A} = \varphi_{1}^{B} = 0
    .
    \label{eq:initial-cond3}
\end{align}

There are a total number of $3L-2$ angles, but only $2L$ independent parameters $a_j$, $b_j$, with $j=1,\dots,L$. It is
useful to compute a constraint for the angles directly, independent from the original parameters of the
Hamiltonian. This constraint would then determine whether there can be a set of parameters $\{a_j,b_j\}_{j=1,\dots,L}$
which would produce a given set of angles.

Taking the appropriate product of the recursion relations we find the constraint
\begin{equation}
  \label{abcrelangles}
  \begin{split}
\tan 2\varphi_{j-1}^{C}
\tan 2\varphi_{j}^{C}
=
\sin 2\varphi_{j}^{A}
 \sin 2\varphi_{j}^{B},\qquad j=2,\dots,L
  \end{split}
\end{equation}
which can also be derived from~\eqref{another-relation-FP}. These relations can be seen as an analogue of \eqref{abcrel}
for the angles of the unitary gates. It can be seen that if 
\eqref{abcrelangles} is satisfied, then there are always some parameters $a_j$, $b_j$ of the Hamiltonian which lead to
these angles via the solution of the recursion relations. The parameter $y$ of the transfer matrix plays the role of an
overall normalization for the parameters $a_j$, $b_j$, and it could be set to $y=1$.

The constraints \eqref{abcrelangles} can be used to specify free fermionic circuits without referring to the original
parameters $a_j$, $b_j$. 
  
It follows from the expression \eqref{Tgenf} that the eigenvalues of our unitary quantum circuit take the form
\begin{equation}
  \prod_{k=1}^{S'} \frac{1\pm i \eps_k}{\sqrt{1+\eps_k^2}},
\end{equation}
where $\eps_k=1/\tilde v_k$ and the $\tilde v_k$ are given by the roots of the polynomials $P_L(v^2)$, as described in
Subsection \ref{sec:genkey}. The normalization of the eigenvalues follows from unitarity, and the difference in the
normalization between formula above and expression \eqref{Tgenf} originates in the normalization \eqref{uj} for the
local gates $u_j^x$ with $x=A, B, C$.

  \section{Conclusions}

\label{sec:concl}

In this work we studied unitary quantum circuits with a hidden free fermionic
spectrum. 
Our central result is that a quantum circuit with 
the
specific structure $G\cdot G^T$ is free fermionic in many regular geometries, both in the original 4-fermion model of
Fendley, and in the generalized model of \cite{sajat-FP-model}. In the case of the linear depth circuits we proved the
free fermionic property analytically, while for the constant depth circuits it was demonstrated numerically. {Our
  results could serve as a starting point for the study of non-equilibrium dynamics generated by the free fermionic circuits.}

Eventually it would be desirable to develop a graph theoretical framework for the unitary circuits, which would generalize the
results of the works \cite{fermions-behind-the-disguise,unified-graph-th} to the current setting. This could be applied to other
models within the class treated in Section \ref{sec:problem}, and not only to the FFD algebra. Ideally, a complete
theory could treat also those models, where there is a certain
amount of fine 
tuning between the coupling constants in the Hamiltonian. An examples for this is the generalized model treated in
Section \ref{sec:FP} in this work.

Our results have some relevance (or perhaps lessons) even to the standard Yang-Baxter integrability. In the previous
works dealing with the hidden free fermions it was often noticed that the systems are integrable in the traditional
sense, and this remains true even with periodic boundary conditions, when the free fermionic properties break
down. Standard methods for integrable models require spatially homogeneous coupling constants, and in those cases new
charges also appear for these models, see \cite{sajat-FP-model}. These connections suggest that the circuits constructed in
this work should be studied and understood also with the methods of Yang-Baxter integrability, perhaps starting with the
periodic cases. It is likely that many regular circuit geometries  will be integrable, even if they are not free
fermionic, and this includes the standard brickwork circuits. It could be useful to better understand the special
structure $\VV=G\cdot G^T$ also from the point of view of Yang-Baxter integrability, perhaps using the framework of
\cite{sajat-medium}.

{A further interesting question is whether there can be quantum circuits with hidden free fermions, such that the
  circuits are still built from local gates, but there is no local Hamiltonian associated with the circuit. We can
not rule out such possibilites, but our current methods are adapted to those cases when there is a local
Hamiltonian in the model, at least for the infinitesimal angle limit. Alternatively, it might also happen to certain circuit geometries become free fermionic for
special choices of the local rotation angles. 
Furthermore, it is a nontrivial question whether there still exists a local Hamiltonian that commutes with the quantum
circuit even when the angles are special finite numbers.
Our treatment does not cover such cases.}

We leave the aforementioned open questions to further work.

\section*{Acknowledgements}

We are thankful to Paul Fendley and Adrian Chapman for inspiring discussions, and 
to Lorenzo Piroli and also Adrian Chapman for communicating their related results. We are also grateful for M\'arton
Mesty\'an for an independent check of our numerical results. Some of the ideas for the computations in Section
\ref{sec:FP} originated from a collaboration with Tam\'as Gombor in the very early stages of this project. 
This work was supported by the
Hungarian National Research, Development and Innovation Office, NKFIH Grant No. K-145904 and the 
NKFIH excellence grant TKP2021$\_$NKTA$\_$64.

\appendix

\section{Free fermionic algebra}

\label{sec:aferm}

Here we discuss the connection between the free fermionic spectrum and the free fermionic operators, together with a certain
gauge freedom in the presence of degeneracies.

Let us assume that a certain Hamiltonian has a free fermionic spectrum in a finite volume $L$. According to the
definition in the main text, this entails that every energy level has the same degeneracy. We denote this degeneracy now
with $D$, and the number of fermionic eigenmodes is denoted by $n$. Let us assume that we have chosen a definite
ordering of the eigenmodes, for example based on the magnitude of the $\eps_k$.

If every energy level is non-degenerate, then the construction of the fermionic operators is obvious. Let us denote the
eigenstates as
\begin{equation}
  \ket{\tau_1,\tau_2,\dots,\tau_n}, \qquad \tau_k=\pm 1,
\end{equation}
where the energy is given by
\begin{equation}
  E=\sum_{k=1}^n \tau_k \eps_k.
\end{equation}
Then the fermionic creation operators $\Psi_k$, $k=1, 2,\dots, n$ can be introduced as
\begin{equation}
  \label{psidef}
  \Psi_k=\mathop{\sum_{\tau_\ell=\pm 1}}_{\ell\ne k} \left(\prod_{\ell=1}^{k-1}\tau_\ell\right)
   \ket{\tau_1,\tau_2,\dots,\tau_k=+1,\dots,\tau_n} \bra{\tau_1,\tau_2,\dots,\tau_k=-1\dots,\tau_n}.
\end{equation}
The corresponding annihilation operator is then $\Psi_{-k}=(\Psi_k)^\dagger$.

This definition ensures the canonical commutation relations (here $k, \ell>0$)
\begin{equation}
  \{\Psi_k,\Psi_\ell\}=0,\qquad   \{\Psi_{-k},\Psi_{-\ell}\}=0,\qquad \{\Psi_k,\Psi_{-\ell}\}=\delta_{k,\ell}
\end{equation}
and the Hamiltonian is then written as
\begin{equation}
H= \sum_{k=1}^n \eps_k \ZZ_k, \qquad \ZZ_k=[\Psi_k,\Psi_{-k}].
\end{equation}

The situation is more involved if there is a uniform degeneracy $D$ in the spectrum. In these cases the fermionic operators
are not well defined. Specific choices can be made if one fixes a basis within each degenerate level. Thus we introduce
the basis vectors
\begin{equation}
  \ket{\tau_1,\tau_2,\dots,\tau_n|a}, \qquad \tau_k=\pm 1,\quad a=1,\dots,D
\end{equation}
for each choice of the signs. Then the fermionic operators can be defined via the formula \eqref{psidef}, supplementing it with a
summation over $a=1,\dots,D$. In such a case the fermionic operators do depend on the choice of the basis, but the
exchange relations are always satisfied. Note however, that the $\ZZ_k$ operators are always well-defined: the
rotations of the basis elements within the respective eigenspaces leave every $\ZZ_k$ invariant.

The freedom of choice of the fermionic operators appeared in the works \cite{fendley-fermions-in-disguise,sajat-FP-model}:
there it entered through the choice for the boundary operator $\chi_{M+1}$. 

\section{Numerical data for circuits in the FFD model}

\label{sec:appnum}

In this Appendix we present a selection of the numerical data that we obtained.
In all the examples below we chose the angles of the gates $u_k$ to be equal to $\varphi_k=\sin(k)$, $k=1, 2,
\dots$. For our purposes these angles behave as random numbers, and it is convenient to fix them at the beginning,
making any future comparisons easier.

Then we compute the numerical eigenvalues $\lambda=e^{i\phi}$ of various circuits, and present their angles
$\phi=-i\log(\lambda)$ with the choice that $\phi\in [-\pi,\pi]$. It was explained in the main text that the eigenvalues
come in complex conjugate pairs, therefore every $\phi$ appears with two different signs.

The differences in the angles are computed by taking the logarithm of the ratios $\lambda_j/\lambda_k$ for all possible
$j,k$. This is equivalent to computing all possible differences of the $\phi$ angles, supplemented with accounting for the
periodicity of $2\pi$ in the differences.

If a unitary operator has $2^n$ unique eigenvalues, and if it is free fermionic, then the differences of the angles will
have $3^n$ unique values. In contrast, if the set of angles is reflection symmetric but otherwise the spectrum is
generic, then there will be $2\cdot 4^{n-1}+1$ unique differences. If the spectrum is the union of multiple free
fermionic spectra (corresponding to sectors determined by symmetries) then the number of unique differences will
typically be between
these two numbers. 

The data below are produced by using the minimal representations discussed in the main text. We considered circuits in
the FFD algebra with $M\le 12$; these circuits can be represented on $L=6$
qubits, therefore the numerical diagonalization can be implemented with minimal effort.
The reader is invited to repeat these numerical computations. 

\bigskip

For the circuit $\VV=u_1u_2u_3u_4$ we find 8 different angles
\begin{equation}
 \pm 0.90790,\ \pm  0.93150,\ \pm   1.47642,\ \pm   1.69880.
\end{equation}
Their differences give 33 possible values, therefore this circuit is found to be not free fermionic. This is in
accordance with the computations in Subsection \ref{sec:M4}.

The reader might wonder whether this circuit and its direct extensions to larger $M$ could be free fermionic in a wider
sense. Perhaps there exist some number of superselection sectors (possibly growing with $M$), such that the model is
free fermionic within each sector. In order to test this hypothese we also consider the spectrum of
$\VV=u_1u_2u_3u_4u_5u_6u_7u_8u_9u_{10}u_{11}u_{12}$. This circuit can be represented on 
6 qubits. We obtained the following $2^6=64$ unique angles:
\begin{equation}
  \begin{split}
&    \pm 0.059883,\ \pm   0.093810,\ \pm    0.131038,\ \pm    0.158502,\ \pm    0.520456,\ \pm    0.599805,\ \pm    0.676978,\ \pm    0.738023  \\
& \pm    0.797634,\ \pm    0.881392,\ \pm    1.122102,\ \pm    1.153706,\ \pm    1.236953,\ \pm    1.256336,\ \pm    1.371849,\ \pm    1.442192\\
 & \pm   1.512413,\ \pm   1.558035,\ \pm    1.716604,\ \pm    1.808851,\ \pm    1.999112,\ \pm    2.152408,\ \pm    2.173039,\ \pm    2.207634\\
 & \pm   2.386399,\ \pm    2.469262,\ \pm    2.501312,\ \pm    2.629893,\ \pm    2.771071,\ \pm    2.812209,\ \pm    3.085452,\ \pm    3.133734\\
  \end{split}
\end{equation}
They give $2\cdot 4^5+1=2049$ unique differences, which is the value for the differences with a generic
reflection symmetric spectrum. This excludes the possibility of having a union of a few number of free fermionic spectra.

Let us now turn to the free fermionic circuits.
For the circuit $\VV=G\cdot G^T$ with $G=u_1u_2u_3u_4$ we find 4 different angles
\begin{equation}
 \pm 0.11319,\  \pm 3.00427.
\end{equation}
As explained in the main text, this is a free fermionic circuit.

In the remaining examples we always consider the structure $\VV=G\cdot G^T$.
In the case of $G=u_1u_2u_3u_4u_5u_6u_7$ we find 8 different angles
\begin{equation}
\pm 1.2386,\   \pm 1.3379,\   \pm 1.8236,\pm   1.8830.
\end{equation}
These values give 27 unique differences, confirming that the circuit is free fermionic.

We also treat the new geometries that we found, and present one numerical example for each geometry.

For the case of $G$ given by \eqref{G1} we consider the specialization to $M=12$, in which case we obtain 16 unique
angles
\begin{equation}
\pm  0.38339,\ \pm   0.47935,\ \pm    0.57221,\ \pm    0.66313,\ \pm    2.47595,\ \pm    2.57190,\ \pm 
   2.66477,\ \pm    2.75568.
\end{equation}
These numbers give 81 unique differences, thus confirming the free fermionic nature.

For $G$ of the form \eqref{G2} we take again $M=12$, and we obtain 16 unique
angles
\begin{equation}
 \pm 0.33711,\ \pm    0.59747,\ \pm    1.06579,\ \pm    1.32615,\ \pm    1.81545,\ \pm    2.07581,\ \pm 
   2.54413,\ \pm    2.80449.
\end{equation}
These numbers also give 81 unique differences, confirming the free fermionic nature.

For $G$ of the form \eqref{G3} we take again $M=12$, and we obtain  16 unique
angles
\begin{equation}
\pm 0.47513,\ \pm   0.69983,\ \pm    0.95615,\ \pm    1.18029,\ \pm    1.96102,\ \pm    2.18572,\ \pm 
   2.44204,\ \pm    2.66618.
\end{equation}
Once again, these numbers give 81 unique differences,  confirming the free fermionic nature.

We also consider the example of $G=u_1u_4u_3u_5u_6u_8$, which was presented in Section \ref{sec:numerics} as one of the
simplest 
circuits with structure $\VV=G\cdot G^T$ which are not free fermionic. We obtain 8 unique angles
\begin{equation}
 \pm 1.0080,\  \pm 1.3259,\ \pm  1.8400,\ \pm   2.1092,
\end{equation}
and these numbers give 33 unique differences. Therefore this circuit is not free fermionic.

\section{Proof of the factorization in the generalized model}

\label{app:FPproof}

\subsection{Commutation relations}
\label{app:FP-commutation}
 
Now we show commutation relations of the operators~\eqref{eq:FP-basic-op}.
The various terms in the Hamiltonian satisfy the following commutation relations, which is the generalization of the FFD algebra~\eqref{ffdalgebra}. 
Two different terms anti-commute in the following cases:
\begin{align}
 \{A_{j},A_{j\pm1}\}=&\{A_{j},A_{j\pm2}\}=0,\\\{A_{j},C_{j\pm1}\}=&\{B_{j},C_{j\pm1}\}=0,\\\{A_{j},C_{j}\}=&\{B_{j},C_{j}\}=0,\\\{A_{j},B_{j\pm1}\}=&0,\\\{A_{j},C_{j-2}\}=&\{B_{j},C_{j-2}\}=0,  
\end{align}
and all the other pairs commute:
\begin{align}
[A_{j},B_{j}] =[A_{j},B_{j\pm2}]=[C_{j},C_{l}] =
[A_{j},C_{j+2}] =[B_{j},C_{j+2}]=0.    
\end{align}

These relations form the basis of the subsequent proof for the factorization of the transfer matrix of the generalized model.

\subsection{Proof of Eq.~\eqref{another-relation-FP}}
\label{app:proof-another-relation}

Here, we prove that the angles defined by the recursion equations~\eqref{eq:recursion-angle-FP-a}-\eqref{eq:recursion-angle-FP-c} indeed satisfy Eq.~\eqref{another-relation-FP}. Equivalently, we will demonstrate the following equation:
\begin{align}
  \label{equivalent-relation}
 \cos\phi_{j}^{A}\cos\phi_{j}^{B}\sin\phi_{j-1}^{C}\sin\phi_{j}^{C}
 =
 \sin\phi_{j}^{A}\sin\phi_{j}^{B}
 .
\end{align}
The proof of~\eqref{equivalent-relation} is as follows:
\begin{align}
    &
 \cos\phi_{j}^{A}\cos\phi_{j}^{B}\sin\phi_{j-1}^{C}\sin\phi_{j}^{C}
 \nonumber\\
 =&
 \cos\phi_{j}^{A}\cos\phi_{j}^{B}\frac{ua_{j-1}b_{j-1}}{\cos^{2}\phi_{j-2}^{C}\cos\phi_{j-2}^{A}\cos\phi_{j-2}^{B}\cos\phi_{j-1}^{A}\cos\phi_{j-1}^{B}}
 \nonumber\\
    &\hspace*{15em}
    \times
 \frac{ua_{j}b_{j}}{\cos^{2}\phi_{j-1}^{C}\cos\phi_{j-1}^{A}\cos\phi_{j-1}^{B}\cos\phi_{j}^{A}\cos\phi_{j}^{B}}
 \nonumber\\
 =&
 \frac{ua_{j-1}a_{j}}{\cos\phi_{j-2}^{C}\cos\phi_{j-1}^{C}\cos\phi_{j-2}^{A}\cos\phi_{j-1}^{A}\cos\phi_{j-1}^{B}}
 \times 
 \frac{ub_{j-1}b_{j}}{\cos\phi_{j-2}^{C}\cos\phi_{j-2}^{B}\cos\phi_{j-1}^{C}\cos\phi_{j-1}^{A}\cos\phi_{j-1}^{B}}
 \nonumber\\
 =&
 \sin\phi_{j}^{A}\sin\phi_{j}^{B}
 ,
\end{align}
where we used the recursion equations~\eqref{eq:recursion-angle-FP-a} and~\eqref{eq:recursion-angle-FP-b} in the first equality, and used the recursion equation~\eqref{eq:recursion-angle-FP-c} in the second equality.

This concludes the proof of Eq.~\eqref{another-relation-FP}.

\subsection{Proof of Eqs.~\eqref{eq:FPfact},~\eqref{eq:FPfact-A} and~\eqref{eq:FPfact-B}}
\label{app:proof-factorization-FP}


To prove the factorization of the transfer matrices~\eqref{eq:FPfact},~\eqref{eq:FPfact-A} and~\eqref{eq:FPfact-B}, we will prove the operator determined from them indeed satisfies the recursion equations of the transfer matrix~\eqref{recursion}-\eqref{recursion-B} derived in~\cite{sajat-FP-model}.

The recursion equation~\eqref{recursion} for $L > 1$ is confirmed as follows:
\begin{align}
  &
 T_{L}-T_{L}^{A}-T_{L}^{B}+T_{L-1} + v C_{L}T_{L-2}
 \nonumber\\
 =&
 G_{L-2}
 \Biggl[
 (g_{L-1}^{A})^{2}(g_{L-1}^{B})^{2}t_{L-1}^{C}t_{L}^{C}\cos\phi_{L}^{A}\cos\phi_{L}^{B}
    \times
 \left(\sin\phi_{L-1}^{C}\sin\phi_{L}^{C}-\tan\phi_{L}^{A}\tan\phi_{L}^{B}\right)
 \nonumber
    \\
    &
 +
 t_{L}^{C}\cos\phi_{L-1}^{A}\cos\phi_{L-1}^{B}\cos\phi_{L}^{A}\cos\phi_{L}^{B}
 \nonumber\\
    &\hspace*{3em}
    \times
 \left(-\sin\phi_{L-1}^{C}\tan\phi_{L}^{A}\tan\phi_{L}^{B}+\sin\phi_{L}^{C}+\frac{v\alpha_{L}\beta_{L}}{\cos\phi_{L-1}^{A}\cos\phi_{L-1}^{B}\cos\phi_{L}^{A}\cos\phi_{L}^{B}}\right)
 \Biggr]G_{L-2}^{\top}
 \nonumber\\
 =&
 G_{L-2}
 \Biggl[
 t_{L}^{C}\cos\phi_{L-1}^{A}\cos\phi_{L-1}^{B}\cos\phi_{L}^{A}\cos\phi_{L}^{B}\cos^2\phi_{L-1}^{C}
 \nonumber
    \\
    &\hspace*{7em}
    \times
 \left(\sin\phi_{L}^{C}+\frac{v\alpha_{L}\beta_{L}}{\cos\phi_{L-1}^{A}\cos\phi_{L-1}^{B}\cos\phi_{L}^{A}\cos\phi_{L}^{B} \cos^2\phi_{L-1}^{C}}\right)
 \Biggr]G_{L-2}^{\top}
 \nonumber\\
 =& 0
 ,
\end{align}
where the first equality follows from expanding the expressions in~\eqref{eq:FPfact},~\eqref{eq:FPfact-A} and~\eqref{eq:FPfact-B} and applying the commutation relations explained in Appendix~\ref{app:FP-commutation}. 
In the second equality, we used Eq.\eqref{another-relation-FP}, and in the final equality, we employed Eq.~\eqref{recursion}.
{For $L=2$, the same equality holds after removing $G_0$ and $G_0^\top$.}
The recursion equation~\eqref{recursion} for the case of $L=1$ is also confirmed as follows:
\begin{align}
 T_{1} - T_{1}^{A} - T_{1}^{B} + T_{0} + v C_{1}T_{-1}
 =&
 (g_1^{C})^2 -1 + v C_{1}
 \nonumber\\
 =&
 ( \sin\phi_{1}^{C} + v a_1 b_1 ) t_{1}^{C} 
 =
  0
 ,
\end{align}
where we used Eq.~\eqref{eq:recursion-angle-FP-c} with $j=0$ in the last equality.
Thus, we have proven that the transfer matrices~\eqref{eq:FPfact},~\eqref{eq:FPfact-A} and~\eqref{eq:FPfact-B} indeed satisfy the defining recursion equation for the transfer matrix~\eqref{recursion}.

The recursion equation~\eqref{recursion-A} for the case of  $L>1$ is confirmed as follows:
\begin{align}
& T_{L+1}^{A} - T_{L} + v A_{L+1} T_{L-1}^{B}
 \nonumber\\
 &= G_{L-2} (g_{L-1}^{B})^{2} t_{L+1}^{A}
 \left[
 \cos\phi_{L-1}^{C} \cos\phi_{L-1}^{A} \cos\phi_{L}^{A}\cos\phi_{L}^{B} \cos\phi_{L}^{C} \sin\phi_{L+1}^{A}
 +
 v a_{L}a_{L+1}
 \right]
 G_{L-2}^{\top}
 \nonumber\\
 &=
  0
 ,
\end{align}
where in the second equality, we used Eq.~\eqref{recursion-A}.
The recursion equation~\eqref{recursion-A} for $L=1$ is confirmed as follows:
\begin{align}
  \nonumber
  T_{2}^{A}-T_{1}+uA_{L+1}T_{0}^{B}&=
 g_1^{C} (g_2^{A})^2 g_1^{C} - (g_1^{C})^2
 \\
&= (\sin\phi_1^C \sin \phi_2^A + v a_1a_2) t_2^A=0,
\end{align}
where we used Eq.~\eqref{eq:recursion-angle-FP-a} with $j=1$ in the last equality.
Thus, we have proven that the transfer matrices~\eqref{eq:FPfact},~\eqref{eq:FPfact-A} and~\eqref{eq:FPfact-B} indeed satisfy the defining recursion equation for the transfer matrix~\eqref{recursion-A}.

In a similar way, we can prove that the factorized transfer matrices~\eqref{eq:FPfact},~\eqref{eq:FPfact-A} and~\eqref{eq:FPfact-B} actually satisfy the recursion equation~\eqref{recursion-B}.

\bigskip


\begin{thebibliography}{10}

\bibitem{jordan-wigner}
P.~Jordan and E.~Wigner, ``{\"U}ber das Paulische {\"A}quivalenzverbot,''
  \href{http://dx.doi.org/10.1007/BF01331938}{{\em Z. Physik} {\bf 47} (1928)
  631--651}.

\bibitem{XX-original}
E.~Lieb, T.~Schultz, and D.~Mattis, ``Two soluble models of an
  antiferromagnetic chain,''
  \href{http://dx.doi.org/10.1016/0003-4916(61)90115-4}{{\em Ann. Phys.} {\bf
  16} (1961) no.~3, 407--466}.

\bibitem{Schulz-Mattis-Lieb}
T.~D. Schultz, D.~C. Mattis, and E.~H. Lieb, ``Two-Dimensional Ising Model as a
  Soluble Problem of Many Fermions,''
  \href{http://dx.doi.org/10.1103/RevModPhys.36.856}{{\em Rev. Mod. Phys.} {\bf
  36} (1964)  856--871}.

\bibitem{chapman-jw}
A.~{Chapman} and S.~T. {Flammia}, ``{Characterization of solvable spin models
  via graph invariants},''
  \href{http://dx.doi.org/10.22331/q-2020-06-04-278}{{\em Quantum} {\bf 4}
  (2020)  278}, \href{http://arxiv.org/abs/2003.05465}{{\tt arXiv:2003.05465
  [quant-ph]}}.

\bibitem{japan-free-fermion-JW}
M.~{Ogura}, Y.~{Imamura}, N.~{Kameyama}, K.~{Minami}, and M.~{Sato},
  ``{Geometric criterion for solvability of lattice spin systems},''
  \href{http://dx.doi.org/10.1103/PhysRevB.102.245118}{{\em Phys. Rev. B} {\bf
  102} (2020) no.~24, 245118}, \href{http://arxiv.org/abs/2003.13264}{{\tt
  arXiv:2003.13264 [cond-mat.stat-mech]}}.

\bibitem{fendley-fermions-in-disguise}
P.~{Fendley}, ``{Free fermions in disguise},''
  \href{http://dx.doi.org/10.1088/1751-8121/ab305d}{{\em J. Phys. A} {\bf 52}
  (2019) no.~33, 335002}, \href{http://arxiv.org/abs/1901.08078}{{\tt
  arXiv:1901.08078 [cond-mat.stat-mech]}}.

\bibitem{alcaraz-medium-fermion-1}
F.~C. {Alcaraz} and R.~A. {Pimenta}, ``{Free fermionic and parafermionic
  quantum spin chains with multispin interactions},''
  \href{http://dx.doi.org/10.1103/PhysRevB.102.121101}{{\em Phys. Rev. B} {\bf
  102} (2020) no.~12, 121101}, \href{http://arxiv.org/abs/2005.14622}{{\tt
  arXiv:2005.14622 [cond-mat.stat-mech]}}.

\bibitem{alcaraz-medium-fermion-2}
F.~C. Alcaraz and R.~A. Pimenta, ``Integrable quantum spin chains with free
  fermionic and parafermionic spectrum,''
  \href{http://dx.doi.org/10.1103/PhysRevB.102.235170}{{\em Phys. Rev. B} {\bf
  102} (2020)  235170}, \href{http://arxiv.org/abs/2010.01116}{{\tt
  arXiv:2010.01116 [cond-mat.stat-mech]}}.

\bibitem{alcaraz-medium-fermion-3}
F.~C. Alcaraz and R.~A. Pimenta, ``Free-parafermionic $Z(N)$ and free-fermionic
  $XY$ quantum chains,''
  \href{http://dx.doi.org/10.1103/PhysRevE.104.054121}{{\em Phys. Rev. E} {\bf
  104} (2021)  054121}, \href{http://arxiv.org/abs/2108.04372}{{\tt
  arXiv:2108.04372 [cond-mat.stat-mech]}}.

\bibitem{rodrigo-ising-and-ffd}
F.~C. {Alcaraz}, R.~A. {Pimenta}, and J.~{Sirker}, ``{Ising analogs of quantum
  spin chains with multispin interactions},''
  \href{http://dx.doi.org/10.1103/PhysRevB.107.235136}{{\em Phys. Rev. B} {\bf
  107} (2023) no.~23, 235136}, \href{http://arxiv.org/abs/2303.15284}{{\tt
  arXiv:2303.15284 [cond-mat.stat-mech]}}.

\bibitem{rodrigo-random-ffd}
F.~C. {Alcaraz}, J.~A. {Hoyos}, and R.~A. {Pimenta}, ``{Random free-fermion
  quantum spin chain with multispin interactions},''
  \href{http://dx.doi.org/10.1103/PhysRevB.108.214413}{{\em Phys. Rev. B} {\bf
  108} (2023) no.~21, 214413}, \href{http://arxiv.org/abs/2308.16249}{{\tt
  arXiv:2308.16249 [cond-mat.dis-nn]}}.

\bibitem{fermions-behind-the-disguise}
S.~J. {Elman}, A.~{Chapman}, and S.~T. {Flammia}, ``{Free fermions behind the
  disguise},'' \href{http://dx.doi.org/10.1007/s00220-021-04220-w}{{\em Commun.
  Math. Phys.} {\bf 388} (2021)  969--1003},
  \href{http://arxiv.org/abs/2012.07857}{{\tt arXiv:2012.07857 [quant-ph]}}.

\bibitem{unified-graph-th}
A.~{Chapman}, S.~J. {Elman}, and R.~L. {Mann}, ``{A Unified Graph-Theoretic
  Framework for Free-Fermion Solvability},'' {\em arXiv e-prints} (2023)  ,
  \href{http://arxiv.org/abs/2305.15625}{{\tt arXiv:2305.15625 [quant-ph]}}.

\bibitem{free-fermion-subsystem-codes}
A.~{Chapman}, S.~T. {Flammia}, and A.~J. {Koll{\'a}r}, ``{Free-Fermion
  Subsystem Codes},'' \href{http://dx.doi.org/10.1103/PRXQuantum.3.030321}{{\em
  PRX Quantum} {\bf 3} (2022) no.~3, 030321},
  \href{http://arxiv.org/abs/2201.07254}{{\tt arXiv:2201.07254 [quant-ph]}}.

\bibitem{free-parafermion}
P.~{Fendley}, ``{Free parafermions},''
  \href{http://dx.doi.org/10.1088/1751-8113/47/7/075001}{{\em J. Phys. A} {\bf
  47} (2014) no.~7, 075001}, \href{http://arxiv.org/abs/1310.6049}{{\tt
  arXiv:1310.6049 [cond-mat.stat-mech]}}.

\bibitem{free-parafermionic-graphs}
R.~L. {Mann}, S.~J. {Elman}, D.~R. {Wood}, and A.~{Chapman}, ``{A
  Graph-Theoretic Framework for Free-Parafermion Solvability},'' {\em arXiv
  e-prints} (2024)  , \href{http://arxiv.org/abs/2408.09684}{{\tt
  arXiv:2408.09684 [quant-ph]}}.

\bibitem{sajat-FP-model}
P.~{Fendley} and B.~{Pozsgay}, ``{Free fermions beyond Jordan and Wigner},''
  \href{http://dx.doi.org/10.21468/SciPostPhys.16.4.102}{{\em SciPost Physics}
  {\bf 16} (2024) no.~4, 102}, \href{http://arxiv.org/abs/2310.19897}{{\tt
  arXiv:2310.19897 [cond-mat.stat-mech]}}.

\bibitem{cooper-anyon}
P.~{Fendley} and K.~{Schoutens}, ``{Cooper pairs and exclusion statistics from
  coupled free-fermion chains},''
  \href{http://dx.doi.org/10.1088/1742-5468/2007/02/P02017}{{\em J. Stat.
  Mech.} {\bf 2007} (2007) no.~2, 02017},
  \href{http://arxiv.org/abs/cond-mat/0612270}{{\tt arXiv:cond-mat/0612270
  [cond-mat.stat-mech]}}.

\bibitem{gyuri-susy-1}
J.~{de Gier}, G.~Z. {Feher}, B.~{Nienhuis}, and M.~{Rusaczonek}, ``{Integrable
  supersymmetric chain without particle conservation},''
  \href{http://dx.doi.org/10.48550/arXiv.1510.02520}{{\em J. Stat. Mech.}
  (2016)  023104}, \href{http://arxiv.org/abs/1510.02520}{{\tt arXiv:1510.02520
  [cond-mat.quant-gas]}}.

\bibitem{gyuri-susy-2}
G.~Z. {Feher}, A.~{Garbali}, J.~{de Gier}, and K.~{Schoutens},
  \href{http://dx.doi.org/10.1007/978-3-030-04161-8_12}{``{A curious mapping
  between supersymmetric quantum chains},''} in {\em 2017 MATRIX Annals},
  pp.~167--184.
\newblock Springer International Publishing, Cham, 2019.
\newblock \href{http://arxiv.org/abs/1710.02658}{{\tt arXiv:1710.02658
  [cond-mat.stat-mech]}}.

\bibitem{viktor-floquet}
V.~Eisler and I.~Peschel, ``Entanglement in a periodic quench,''
  \href{http://dx.doi.org/10.1002/andp.200810299}{{\em Annalen der Physik} {\bf
  17} (2008) no.~6, 410–423}, \href{http://arxiv.org/abs/0803.2655}{{\tt
  arXiv:0803.2655 [cond-mat.stat-mech]}}.

\bibitem{prosen-150}
J.~W.~P. {Wilkinson}, T.~{Prosen}, and J.~P. {Garrahan}, ``{Exact solution of
  the ``Rule 150'' reversible cellular automaton},''
  \href{http://dx.doi.org/10.1103/PhysRevE.105.034124}{{\em Phys. Rev. E.} {\bf
  105} (2022) no.~3, 034124}, \href{http://arxiv.org/abs/2110.15085}{{\tt
  arXiv:2110.15085 [cond-mat.stat-mech]}}.

\bibitem{kicked-ising-eredeti}
T.~{Prosen}, ``{Exact Time-Correlation Functions of Quantum Ising Chain in a
  Kicking Transversal Magnetic Field ---Spectral Analysis of the Adjoint
  Propagator in Heisenberg Picture---},''
  \href{http://dx.doi.org/10.1143/PTPS.139.191}{{\em Prog. Theor. Phys. Supp.}
  {\bf 139} (2000)  191--203}, \href{http://arxiv.org/abs/nlin/0009031}{{\tt
  arXiv:nlin/0009031 [nlin.CD]}}.

\bibitem{kicked-Ising-space-time-duality-1}
M.~Akila, D.~Waltner, B.~Gutkin, and T.~Guhr, ``Particle-time duality in the
  kicked Ising spin chain,''
  \href{http://dx.doi.org/10.1088/1751-8113/49/37/375101}{{\em J. Phys. A} {\bf
  49} (2016) no.~37, 375101}.

\bibitem{dual-unitary-1}
B.~{Bertini}, P.~{Kos}, and T.~{Prosen}, ``{Exact Spectral Form Factor in a
  Minimal Model of Many-Body Quantum Chaos},''
  \href{http://dx.doi.org/10.1103/PhysRevLett.121.264101}{{\em Phys. Rev.
  Lett.} {\bf 121} (2018) no.~26, 264101},
  \href{http://arxiv.org/abs/1805.00931}{{\tt arXiv:1805.00931 [nlin.CD]}}.

\bibitem{dual-unitary-2}
B.~{Bertini}, P.~{Kos}, and T.~{Prosen}, ``{Entanglement Spreading in a Minimal
  Model of Maximal Many-Body Quantum Chaos},''
  \href{http://dx.doi.org/10.1103/PhysRevX.9.021033}{{\em Phys. Rev. X} {\bf 9}
  (2019) no.~2, 021033}, \href{http://arxiv.org/abs/1812.05090}{{\tt
  arXiv:1812.05090 [cond-mat.stat-mech]}}.

\bibitem{sajat-ffd-corr}
I.~{Vona}, M.~{Mesty{\'a}n}, and B.~{Pozsgay}, ``{Exact real time dynamics with
  free fermions in disguise},'' {\em arXiv e-prints} (2024)  ,
  \href{http://arxiv.org/abs/2405.20832}{{\tt arXiv:2405.20832
  [cond-mat.stat-mech]}}.

\bibitem{integrable-trotterization}
M.~{Vanicat}, L.~{Zadnik}, and T.~{Prosen}, ``{Integrable Trotterization: Local
  Conservation Laws and Boundary Driving},''
  \href{http://dx.doi.org/10.1103/PhysRevLett.121.030606}{{\em Phys. Rev.
  Lett.} {\bf 121} (2018) no.~3, 030606},
  \href{http://arxiv.org/abs/1712.00431}{{\tt arXiv:1712.00431
  [cond-mat.stat-mech]}}.

\bibitem{matchgates1}
L.~G. Valiant, ``Quantum Circuits That Can Be Simulated Classically in
  Polynomial Time,'' \href{http://dx.doi.org/10.1137/S0097539700377025}{{\em
  SIAM Journal on Computing} {\bf 31} (2002) no.~4, 1229--1254}.

\bibitem{matchgates2}
B.~M. {Terhal} and D.~P. {Divincenzo}, ``{Classical simulation of
  noninteracting-fermion quantum circuits},''
  \href{http://dx.doi.org/10.1103/PhysRevA.65.032325}{{\em Phys. Rev. A} {\bf
  65} (2002) no.~3, 032325}, \href{http://arxiv.org/abs/quant-ph/0108010}{{\tt
  arXiv:quant-ph/0108010 [quant-ph]}}.

\bibitem{matchgates3}
R.~{Jozsa} and A.~{Miyake}, ``{Matchgates and classical simulation of quantum
  circuits},'' \href{http://dx.doi.org/10.1098/rspa.2008.0189}{{\em Proc. Roy.
  Soc. Lond. A} {\bf 464} (2008) no.~2100, 3089--3106},
  \href{http://arxiv.org/abs/0804.4050}{{\tt arXiv:0804.4050 [quant-ph]}}.

\bibitem{holocode-review-jahn-eisert}
A.~{Jahn} and J.~{Eisert}, ``{Holographic tensor network models and quantum
  error correction: a topical review},''
  \href{http://dx.doi.org/10.1088/2058-9565/ac0293}{{\em Quantum Sci. Tech.}
  {\bf 6} (2021) no.~3, 033002}, \href{http://arxiv.org/abs/2102.02619}{{\tt
  arXiv:2102.02619 [quant-ph]}}.

\bibitem{kitaev-honeycomb}
A.~{Kitaev}, ``{Anyons in an exactly solved model and beyond},''
  \href{http://dx.doi.org/10.1016/j.aop.2005.10.005}{{\em Ann. Phys.} {\bf 321}
  (2006) no.~1, 2--111}, \href{http://arxiv.org/abs/cond-mat/0506438}{{\tt
  arXiv:cond-mat/0506438 [cond-mat.mes-hall]}}.

\bibitem{znidaric-classification}
U.~Duh and M.~Znidaric, ``Classification of same-gate quantum circuits and
  their space-time symmetries with application to the level-spacing
  distribution,'' {\em arXiv e-prints} (2024)  ,
  \href{http://arxiv.org/abs/2401.09708}{{\tt arXiv:2401.09708 [quant-ph]}}.

\bibitem{sajat-medium}
T.~{Gombor} and B.~{Pozsgay}, ``{Integrable spin chains and cellular automata
  with medium-range interaction},''
  \href{http://dx.doi.org/10.1103/PhysRevE.104.054123}{{\em Phys. Rev. E} {\bf
  104} (2021) no.~5, 054123}, \href{http://arxiv.org/abs/2108.02053}{{\tt
  arXiv:2108.02053 [nlin.SI]}}.

\end{thebibliography}

\providecommand{\href}[2]{#2}\begingroup\raggedright\endgroup

\end{document}